\ifdefined\pdfoutput
  \pdfoutput=1
  \documentclass[12pt]{article}
\else
  \documentclass[12pt,dvipdfmx]{article}
\fi

\ifdefined\pdfoutput
  \usepackage{xcolor}
\else
  \usepackage[dvipdfmx]{xcolor}
\fi
\usepackage{mathrsfs,bm,sasaki,braket}
 
\usepackage{mathtools,amssymb,amsthm}   

\usepackage{tikz}  

\usepackage[a4paper]{geometry}
\newgeometry{left=30mm,right=30mm,bottom=32mm,top=32mm}

\usepackage[unicode,psdextra,hypertexnames=false]{hyperref}
\hypersetup{
    colorlinks=true,
    linkcolor=blue!60!black,  
    urlcolor=teal!80!black,   
    citecolor=green!50!black  
}

\usepackage{autonum}

\newcommand{\Phig}{\Phi_\mathrm{g}}
\newcommand{\Phigt}{{\til\Phi}_\mathrm{g}}

\theoremstyle{plain}
\newtheorem{Theorem}{Theorem}[section]

\newtheorem{Corollary}[Theorem]{Corollary}
\newtheorem{Lemma}[Theorem]{Lemma}
\newtheorem{Proposition}[Theorem]{Proposition}
\newtheorem{Conjecture}[Theorem]{Conjecture}
\numberwithin{equation}{section}

\theoremstyle{remark}
\newtheorem{Remark}[Theorem]{Remark}
\newtheorem{Example}[Theorem]{Example}

\title{\sc Holomorphy of the ground state and energy for bosonic quadratic Hamiltonians}

\author{%
  Kota Imura \and
  Shinnosuke Izumi \and
  Yasumichi Matsuzawa \and
  Itaru Sasaki
}

\date{}   

\begin{document}
\maketitle   

\begin{abstract}
We study the holomorphy of the ground state and its energy with respect to the coupling constant
for bosonic quadratic Hamiltonians.
We prove general theorems establishing this holomorphy, thereby providing a 
rigorous justification for formal perturbative expansions. 
This theory is applied to several concrete models, including the single pair interaction 
model and the Pauli--Fierz model in the dipole approximation. 
Importantly, for the Pauli--Fierz model with a realistic ultraviolet cutoff, the perturbation expansions 
in the elementary charge converge at the physical charge value, even when mass renormalization is taken into account.
Furthermore, we explicitly determine the radii of convergence 
of the ground states and their energies for both the single pair interaction model
 and the fibers of the translation-invariant Pauli--Fierz model.
\end{abstract}

\tableofcontents 


\section{Introduction}
Many discussions in quantum field theory rely on perturbative expansions. 
Perturbation theory allows us to expand various quantities as formal power series
 in terms of the coupling constant, so that each term can be determined algebraically.
However, the convergence of such series is generally far from obvious; in many cases, they are known to diverge 
 (see e.g., \cite{Jaffe1965}). 
Therefore, a rigorous understanding of the regime in which formal perturbative
 expansions are justified remains a fundamental challenge in the mathematical study of quantum fields.

 For models coupled to a massive bosonic field, analyticity at small couplings often follows from 
Kato's analytic perturbation theory \cite{MR1335452}. By contrast, the massless case is substantially more delicate.
To address this, several rigorous approaches have been developed.
 For instance, Bach, Fröhlich, and Sigal \cite{BFS1998} introduced the operator-theoretic 
renormalization group method to construct ground states. 
Based on this framework, Griesemer and Hasler \cite{MR2519822} proved the 
analyticity of ground states and their energies.
There have also been various related studies 
\cite{MR2846670, MR2890307, Arai2014, MR3926124, MR3339159, MR2738099, HaslerHerbst2011_SB, HaslerHerbst2011_Convergent, HaslerHerbst2011_Smooth, hasler2025resonances}.
However, these results are typically restricted to the regime of small coupling constants. 
Establishing analyticity for larger coupling and determining the radius of
 convergence---and whether it reaches actual physical values---remain open problems.

In this paper, we consider bosonic quadratic Hamiltonians. Under appropriate assumptions, 
such Hamiltonians can be diagonalized via Bogoliubov transformations, 
and it is known that the ground states and their energies can be expressed explicitly 
in terms of one-particle operators \cite{MR4213757, Rui78}.
More precisely, the ground states are given by squeezed states, 
while the energies are represented as traces of differences of 
certain one-particle operators.
Thus, an advantage of these models over general quantum field models is that
the structure of the ground states and their energies can be understood relatively explicitly.
 Alternative mathematical approaches to the diagonalization of quadratic Hamiltonians can also be 
found in \cite{BachBru2016, 2017Derezinski, 2016NamNapiSolovej}.

However, because these expressions involve the square roots of unbounded operators,
 their analyticity with respect 
to the coupling constant is not obvious. Furthermore, depending on the value of the coupling constant, 
the Hamiltonian may no longer be bounded from below (see, e.g., \cite{AF}).
This suggests that the ground states and their energies have a finite radius of convergence. 
Therefore, the problem of analyticity is nontrivial.

A naive way to prove the convergence of a perturbative series is to estimate each term 
and show that the series converges. However, in such an approach, errors in the estimates
 tend to accumulate, leading to a bound on the radius of convergence that is far from optimal.
 To avoid this difficulty, we instead extend the coupling constant to a complex parameter
 and establish the holomorphy of the ground states and their energies. 
The convergence of the perturbative expansion then follows immediately from this holomorphy.

In concrete applications, the Hamiltonian does not always depend linearly 
on the coupling constant.
We therefore introduce an abstract framework and study holomorphy in this setting.
Specifically, for the ground state energy, we formulate conditions under which 
the one-particle operators appearing in the trace representation can be defined for 
complex values of the coupling constant, so that the ground state energy admits an analytic continuation.
Under the same conditions, we show that the squeezed state for complex values of 
the parameter is well-defined as a vector in Fock space and is holomorphic in the complex domain.

As applications of this general framework, we consider the single pair interaction model,
 the Pauli--Fierz model in the dipole approximation with a harmonic potential, 
and the fiber Hamiltonians of its translation-invariant version, and provide explicit 
domains in the complex plane to which the ground states and their energies can be 
holomorphically extended.
In particular, for the single pair interaction model and the fibers of the dipole-approximated 
Pauli--Fierz model, we determine the radii of convergence 
for the ground states and their energies. 
Furthermore, for these Pauli--Fierz Hamiltonians, under a physical ultraviolet 
cutoff---namely the cutoff at the Compton wave number used in the explanation of 
the Lamb shift \cite{PhysRev.72.339}---we prove that the Maclaurin series for the 
ground states and their energies with respect to the elementary charge converge at the physical value,
even when the effect of mass renormalization is taken into account.


The rest of this paper is organized as follows. 
In Section \ref{sec:Definition_of_the_Model}, 
we define bosonic quadratic Hamiltonians and summarize their diagonalization
 via Bogoliubov transformations, along with the representations of the ground states 
and their energies.
In Section \ref{sect:GSE}, we present our general setting and, under this setting, 
prove the holomorphy of the ground state energy using integral representations 
of the square roots of m-accretive operators and estimates in the trace norm. 
In Section \ref{sect:GS}, we prove the holomorphy of the ground state based on norm 
and number estimates for squeezed states. 
In Section \ref{sect:example}, we apply this general theory to concrete models: the single pair interaction model, 
the Pauli--Fierz model in the dipole approximation with a harmonic potential, 
and the translation-invariant Pauli--Fierz model in the dipole approximation. 
We remark that, for the translation-invariant model with $\bP\neq \boldsymbol{0}$, 
the ground state is not a simple squeezed state but is accompanied by a displacement,
 requiring additional arguments.

The operator theoretic tools required for the proofs are mainly collected in the appendices. 
In Appendix \ref{AppenA}, we derive a Heinz-type inequality for m-accretive operators. 
We use this to prove that the squeezed state for complex values of the parameter belongs 
to the Fock space. 
In Appendix \ref{AppenB}, we provide a criterion for the holomorphy of Hilbert space-valued
 functions. 
In Appendix \ref{AppenC}, we show that every squeezed state associated with a Hilbert--Schmidt 
operator can be realized as the vacuum of a suitably constructed Bogoliubov transformation---a general, 
model-independent correspondence that may be of independent interest beyond the present application. 
Several auxiliary estimates for squeezed states, such as norm and inner-product formulas, 
are also established there and used throughout the paper.
In Appendix \ref{AppenD}, 
we collect several technical estimates for integrals of trace class operators.
Finally, in Appendix \ref{appenE}, we recall a useful result due to Mattner \cite{MR1823156}
on complex differentiation under the integral sign.

\section{Definition of the Model }\label{sec:Definition_of_the_Model}
In this section, we introduce bosonic quadratic Hamiltonians and provide specific 
expressions for their ground states and ground state energies.
We use the same symbols and definitions as in the previous paper \cite{MR4213757}.
For more fundamental background on Fock spaces, we refer the reader to \cite{AraiBook2}.

Let $\sH$ be a separable complex Hilbert space.
 The boson Fock space over $\sH$ is denoted by $\Fb(\sH)$.
A vector $\Psi \in \Fb(\sH)$ is denoted by $\Psi = (\Psi^{(n)})_{n=0}^\infty$ with
$\Psi^{(n)} \in \tensor_\mathrm{s}^n \sH$.
The standard creation and annihilation operators for $f\in \sH$ are denoted by $A^*(f)$ and $A(f)$,
respectively.
Note that $A^*(f)$ is linear in $f$, and $A(f)$ is anti-linear in $f$.
The Segal field operator $\PhiS(f)$ is defined by
\begin{equation}
  \PhiS(f) \coloneqq  \frac{1}{\sqrt{2}} \ovl{ (A(f)+A^*(f)) },
\end{equation}
where $\ovl{S}$ denotes the closure of a closable operator $S$.
Let $T$ be a self-adjoint operator on $\sH$. 
We denote by $\dGb(T)$ the second quantization of $T$ which is a self-adjoint operator on $\Fb(\sH)$.

The bosonic quadratic Hamiltonian that we consider in this paper is defined by 
\begin{equation}
  H = \dGb(T) + \frac{1}{2} \sum_{n=1}^\infty \lambda_n \PhiS(g_n)^2   \label{HAMIL}
\end{equation}
which acts on $\Fb(\sH)$. 

We assume the following conditions.
\begin{enumerate}
\item[(B1)] $T$ is an injective non-negative self-adjoint operator acting on $\sH$~(i.e., $T>0$).
\item[(B2)] $\lambda_n\in \RR$ and $g_n \in \dom(T^{1/2}) \cap \dom(T^{-1/2})$ for all $n\in\NN$.
\item[(B3)] $\sum_{n=1}^\infty |\lambda_n| \norm{T^{-1/2}g_n}^2<\infty$.
\item[(B4)] $\sum_{n=1}^\infty |\lambda_n| \norm{T^{1/2}g_n}^2<\infty$.
\item[(B5)] For some $\vep>0$, the operator inequality
  \begin{equation}
     \one + \sum_{n=1}^\infty \lambda_n \ketbra{T^{-1/2}g_n}{T^{-1/2}g_n} \geq \vep
  \end{equation}
holds.
\item[(B6)] There exists a conjugation $J$ on $\sH$ such that
  \begin{equation}
    JTJ = T, \qquad Jg_n = g_n, \qquad n\in\NN.
  \end{equation}
\end{enumerate}

We note that $T>0$ does not necessarily mean $\inf\sigma(T)>0$.
The self-adjointness of $H$ is ensured by the following theorem.  
\begin{Theorem}[{\cite[Theorem 4.3]{MR4213757}}] \label{sa}
  Assume (B1)--(B6). Then $H$ is self-adjoint on $\dom(\dGb(T))$, bounded from below,
and essentially self-adjoint on any core for $\dGb(T)$.
\end{Theorem}

Next we explain the results on the diagonalization of the Hamiltonian.
Let 
\begin{equation}
  W \coloneqq  \sum_{n=1}^\infty \lambda_n \ketbra{T^{1/2}g_n}{T^{1/2}g_n}. \label{defW}
\end{equation}
Since $T^2 + W>0$ by (B5), we can define
\begin{equation}
  S \coloneqq  \left( T^2 + W \right)^{1/2}.
\end{equation}
It is shown in \cite[Lemma 5.2]{MR4213757} that
\begin{equation}
  X \coloneqq  \frac{1}{2} \left(\ovl{T^{-1/2}S^{1/2}} + \ovl{T^{1/2}S^{-1/2}} \right), \qquad 
  Y \coloneqq  \frac{1}{2} \left(\ovl{T^{-1/2}S^{1/2}} - \ovl{T^{1/2}S^{-1/2}} \right)  \label{defXY}
\end{equation}
are bounded, and $Y$ is Hilbert--Schmidt on $\sH$.
For $f\in \sH$, we define
\begin{equation}
  B(f) := \ovl{A(Xf)+ A^*(JYf)}.
\end{equation}
Then, the unitary operator $\sU $ on $\Fb(\sH)$ satisfying the equation \eqref{defU*} 
in Appendix \ref{AppenC} yields
\begin{equation}
  \sU B(f)\sU^* = A(f),\qquad f\in \sH.  \label{BtoA}
\end{equation}
Here, $\sU$ provides the unitary implementation of the Bogoliubov transformation.

The bosonic quadratic Hamiltonian $H$ is diagonalized as follows.
\begin{Theorem}[{\cite[Theorems 5.3 and 6.1]{MR4213757}}] \label{bogo}
Assume (B1)--(B6). Then
\begin{equation}
  \sU H \sU^* = \dGb(S) + E,   \label{diag}
\end{equation}
where the ground state energy $E$ is given by
\begin{equation}
  E = \frac{1}{2} \tr (\ovl{S-T}).   \label{gse}
\end{equation}
\end{Theorem}

Since the Fock vacuum $\Omega=(1,0,0,\cdots)$ is the unique ground state of $\dGb(S)$,
 the formula \eqref{diag} implies that 
the Hamiltonian $H$ has a unique ground state
\begin{equation}
  \Phig \coloneqq  \sU^* \Omega.  \label{norgs}
\end{equation}
Using the explicit expression of $\sU$ in terms of $X$ and $Y$ in 
\eqref{defXY} (see also \eqref{defU*} in Appendix \ref{AppenC}),
 the normalized ground state can be written as
\begin{equation}
  \Phig = \det(\one - K_1^*K_1)^{1/4} e^{-\frac{1}{2}\Delta^*(K_1)} \Omega, 
  \label{ngs}
\end{equation}
where $K_1$ is defined in equation \eqref{def:Ks} in Appendix \ref{AppenC},
and $\Delta^*(K)$ denotes the two-particle creation operator associated with a Hilbert--Schmidt 
operator $K$ (see Appendix \ref{AppenC} for details).
Thus, the ground state $\Phig$ has the structure of a squeezed state.

When considering a parameter-dependent Hamiltonian, it is natural to investigate the analyticity 
of the ground state with respect to the parameter.
From the viewpoint of justifying perturbative calculations, 
it is appropriate to consider the following expression, obtained by removing the normalization 
constant from \eqref{ngs}:
\begin{equation}
\til{\Phi}_\mathrm{g}
 \coloneqq  e^{-\frac{1}{2}\Delta^*(K_1)} \Omega 
 = \sum_{n=0}^\infty \left(-\frac{1}{2}\right)^n \frac{1}{n!} (\Delta^*(K_1))^n \Omega.   \label{phitil}
\end{equation}
Indeed, standard formal perturbation theory typically addresses the state 
$\Phigt$, satisfying $\langle \Omega, \Phigt\rangle = 1$, 
rather than the normalized ground state $\Phig$.

Finally, we state a key proposition used to prove the holomorphy of the ground state.
By condition (B6) and \eqref{defXY}, $X$ and $Y$ commute with $J$, and hence
\begin{equation}
   K_1= YX^{-1}.
\end{equation}
To analyze $K_1$, we introduce a simpler operator $K_4$.

\begin{Proposition}\label{prop:K1_K4} 
Set
\begin{equation}
 K_4 \coloneqq   \frac{1}{2} \, \ovl{T^{-1/2}(S-T)T^{-1/2}}.
\end{equation}
Then
\begin{equation}
K_1 =  K_4( 1 + K_4)^{-1}.
\end{equation}  
\end{Proposition}

\begin{proof}
Set $a\coloneqq  \ovl{T^{-1/2}S^{1/2}}$ and $b \coloneqq  \ovl{T^{1/2}S^{-1/2}}$.
Then $X=2^{-1}(a+b)$, $Y=2^{-1}(a-b)$, and 
\begin{align}
  K_1 & = YX^{-1} = (a-b)(a+b)^{-1} = (a-b)b^{-1}(ab^{-1}+1)^{-1} \\
     & = (ab^{-1}-1)(1+ab^{-1})^{-1}.
\end{align}
Since $ab^{-1}-1 = \ovl{T^{-1/2}ST^{-1/2}}-1 = 2K_4$, we have
\begin{equation}
  K_1 = 2K_4(2+2K_4)^{-1} = K_4(1+K_4)^{-1}.
\end{equation}
\end{proof}

\section{Holomorphy of the Ground State Energy}\label{sect:GSE}
Let $\sH $ be a separable complex Hilbert space.
The symbol $\cB(\sH )$ denotes the Banach space consisting of everywhere defined bounded operators on $\sH $. 
Similarly, the symbol $\mathcal{S}_p(\sH )$ denotes the Banach space consisting of $p$-Schatten operators on $\sH $. 
For $A\in\mathcal{S}_p(\sH )$, $\|A\|_p$ denotes the Schatten $p$-norm of $A$.
In particular, it coincides with the trace norm for $p=1$ and with the Hilbert--Schmidt norm for $p=2$.

Recall that a (possibly unbounded) linear operator $A$ acting on $\sH $ is called
\begin{itemize}
  \item \textit{accretive} if 
    \begin{equation}
      \Re \langle f,Af\rangle\geq 0,\qquad \forall f\in \dom(A),
    \end{equation}
  \item \textit{m-accretive} if  $A$ is accretive, closed, and $A+\lambda$ is bijective 
    for some positive constant $\lambda>0$.
\end{itemize}
Note that an everywhere defined bounded operator is m-accretive if and only if it is accretive.

To state our main results, we introduce the following conditions:
\begin{itemize}
 \item[(C1)] $T$ is an injective non-negative self-adjoint operator acting on $\sH $.
 \item[(C2)] $W:U\to\mathcal{S}_1(\sH )$ is a holomorphic function defined on a non-empty open subset $U$ of $\CC$.
 \item[(C3)] For any $z\in U$, the range of $W(z)T^{-1}$ is contained in the domain of $T^{-1}$.
 \item[(C4)] For any $z\in U$, the operator $T^{-1}W(z)T^{-1}$ is bounded, and its closure is of trace class.
 \item[(C5)] The map
  \begin{equation}
      U\to\mathcal{S}_1(\sH ), \qquad z\mapsto\ovl{T^{-1}W(z)T^{-1}}
  \end{equation}
    is holomorphic.
  \item[(C6)] For any $z\in U$, the operator inequality $\Re( 1 +\ovl{T^{-1}W(z)T^{-1}}) >0$ holds.
\end{itemize}

We note that condition (C6) is equivalent to saying that, for any $z\in U$,
 there exists a constant $\delta(z)>0$ such that the operator 
$1+\ovl{T^{-1}W(z)T^{-1}}-\delta(z)$ is accretive.
\begin{Proposition}\label{prop:m-acc}
Assume conditions (C1)--(C6). Then the operator
\begin{equation}   T^2+W(z)   \end{equation}
is m-accretive for any $z\in U$.
\end{Proposition}
\begin{proof}
Let $z\in U$.
By condition (C6), we have
\begin{align}
 \Re \inner{u}{(T^2+W(z))u} 
 & = \Re \inner{Tu}{( \one + T^{-1}W(z)T^{-1}-\delta(z))Tu} + \delta(z)\norm{Tu}^2 \\
 & \geq \delta(z) \norm{Tu}^2,
\end{align}
for all $u \in \dom(T^2)$. This means that $T^2+W(z)$ is an accretive operator.

Since $W(z)$ is bounded, $T^2+W(z)$ is closed. 
Thus it is enough to show that $-\lambda$ is in the resolvent set of $T^2+W(z)$ for some $\lambda>0$.
We consider the operator
\begin{equation}
  T^2 + W(z) -(-\lambda) = ( \one + W(z) (T^2+\lambda)^{-1})(T^2+\lambda).
\end{equation}
Since $T^2$ is non-negative and self-adjoint, $T^2+\lambda$ is a bijection from $\dom(T^2)$ onto $\sH$.
For sufficiently large $\lambda>0$, we have
\begin{equation}
  \norm{W(z)(T^2+\lambda)^{-1}} \leq \norm{W(z)}\lambda^{-1} < 1  .
\end{equation}
Thus, $\one + W(z)(T^2+\lambda)^{-1}$ is a bijection on $\sH$.
Hence, $-\lambda$ is in the resolvent set of $T^2+W(z)$, and the operator $T^2+W(z)$ is m-accretive.
\end{proof}

Since $T^2+W(z)$ is m-accretive, we can consider its square root:
\begin{equation}
   S(z)\coloneqq  \sqrt{T^2+W(z)},\qquad z\in U.
\end{equation}
See \cite[Theorem 3.35, p.~281]{MR1335452} for the definition and general properties 
of the square root of an m-accretive operator.
$S(z)$ is the unique m-accretive operator such that $S(z)^2 = T^2 + W(z)$.
It satisfies
\begin{align}
 S(z) u & = \frac{1}{\pi} \int_0^\infty \lambda^{-1/2}(T^2+W(z)+\lambda)^{-1}(T^2+W(z)) u \, d\lambda\\
        & = \frac{2}{\pi} \int_0^\infty (T^2+ W(z) + t^2)^{-1}(T^2+W(z)) u \, dt,  \label{int_rep_S(z)}
\end{align}
for $u \in \dom(T^2)$.

The following is the main result of this section:
\begin{Theorem}\label{main1}
Assume conditions (C1)--(C6). Then the operator 
\begin{equation}
   (S(z)-T)|_{\dom(T^2)}
\end{equation}
is bounded, and its closure is of trace class for any $z\in U$.
Moreover, the maps
\begin{equation}
   U \to \mathcal{S}_1(\sH ), \qquad z \mapsto \ovl{(S(z)-T)|_{\dom(T^2)}}
\end{equation}
and
\begin{equation}
   U\to\CC,\qquad z\mapsto \tr\Big(\ovl{(S(z)-T)|_{\dom(T^2)}}\Big)
\end{equation}
are holomorphic.
\end{Theorem}

In the following, we prove Theorem \ref{main1}.
To simplify notation, we introduce
\begin{equation}
   R_t \coloneqq  (T^2+t^2)^{-1},\qquad t>0.
\end{equation}

\begin{Lemma}\label{main1 lemma1}
Assume conditions (C1)--(C6).
Let $z\in U$ be arbitrary, and let $\delta(z)\in(0,1]$ be a constant 
for which $1+\ovl{T^{-1}W(z)T^{-1}}-\delta(z)$ is accretive.
Then for any $t>0$, the operator $1+R_t^{1/2}W(z)R_t^{1/2}$ is bijective with
\begin{equation}
  \|(1+R_t^{1/2}W(z)R_t^{1/2})^{-1}\|\leq  \delta(z)^{-1}.
\end{equation}
\end{Lemma}
\begin{proof}
We first show that the operator 
\begin{equation}
  A(z)\coloneqq  1+R_t^{1/2}W(z)R_t^{1/2}-\delta(z)
\end{equation}
is accretive.
Let $f\in\sH $ be arbitrary, and let $T=\int_0^\infty\lambda\, d E_T(\lambda)$ be the spectral resolution of $T$.
Note that
\begin{equation}
   A(z) + \delta(z) = 1+R_t^{1/2}W(z)R_t^{1/2} = 1-T^2R_t+TR_t^{1/2} \big(1+\ovl{T^{-1}W(z)T^{-1}}\big) TR_t^{1/2}.
\end{equation}
Hence, we have
\begin{align}
 & \Re \left\langle f,\big(1+R_t^{1/2}W(z)R_t^{1/2}\big) f\right\rangle
  \geq  \left\langle f,\left(1-T^2R_t+\delta(z) T^2R_t\right)f\right\rangle\\
 &\qquad=\int_0^\infty\frac{t^2+\delta(z)\lambda^2}{\lambda^2+t^2}\, d \|E_T(\lambda)f\|^2
    \geq \int_0^\infty\frac{\delta(z)t^2+\delta(z)\lambda^2}{\lambda^2+t^2}\, d \|E_T(\lambda)f\|^2\\
  &\qquad=\delta(z)\|f\|^2.
\end{align}
Thus $A(z)$ is accretive.
	
Since $A(z)$ is an everywhere defined bounded operator, $A(z)$ is m-accretive as well.
It follows from the general theory of m-accretive operators that $-\delta(z)$ is 
in the resolvent set of $A(z)$ with $ \left\|(A(z)+\delta(z))^{-1}\right\|\leq \delta(z)^{-1}$
(see, e.g., \cite[Proposition 3.20]{MR2953553}).
This completes the proof.
\end{proof}

\begin{Lemma}\label{main1 lemma2}
Assume conditions (C1)--(C6).
Let $K$ be a non-empty compact subset of $\CC$ contained in $U$.
Then there exists a positive constant $\delta\in(0,1]$ such that for all $z\in K$ and $t>0$,
 the operator $1+R_t^{1/2}W(z)R_t^{1/2}$ is bijective with
\begin{equation}
   \|(1+R_t^{1/2}W(z)R_t^{1/2})^{-1}\| \leq  \delta^{-1}.
\end{equation}
\end{Lemma}

\begin{proof}
For each $z\in U$, we set $A(z) \coloneqq  \Re (1+\ovl{T^{-1}W(z)T^{-1}})$ and 
\begin{equation}
   E_0(z)\coloneqq \inf\{\langle f,A(z)f\rangle\mid f\in\sH ,\ \|f\|=1\},\qquad z\in U.
\end{equation}
By condition (C6), each $E_0(z)$ is positive.
A standard argument shows that
\begin{equation}
   |E_0(z_1)-E_0(z_2)| \leq \|A(z_1)-A(z_2)\|, \qquad \forall z_1,z_2\in U.
\end{equation}
 This, together with condition (C5), implies that $E_0(z)$ is continuous on $U$.
Define
\begin{equation}
   \delta_0 \coloneqq  \min_{z\in K} E_0(z) > 0,  \qquad  \delta\coloneqq \min\{\delta_0,1\} \in (0,1].
\end{equation}
Then $1+\ovl{T^{-1}W(z)T^{-1}}-\delta$ is accretive for all $z\in K$.
This, together with Lemma \ref{main1 lemma1}, completes the proof.
\end{proof}

\begin{Lemma}\label{main1 lemma3}
Assume conditions (C1)--(C6).
For any $z\in U$ and $t>0$, we have
\begin{equation}
  (S(z)^2+t^2)^{-1} - (T^2+t^2)^{-1} = -R_t^{1/2}\Big(1+R_t^{1/2} W(z) R_t^{1/2}\Big)^{-1} R_t^{1/2} W(z) R_t.
\end{equation}
\end{Lemma}

\begin{proof}
 Since
\begin{equation}
    S(z)^2+t^2=(T^2+t^2)^{1/2}\Big(1+R_t^{1/2}W(z)R_t^{1/2}\Big)(T^2+t^2)^{1/2},   
\end{equation}
Lemma \ref{main1 lemma1} yields that
\begin{equation}
  (S(z)^2+t^2)^{-1} = R_t^{1/2}\Big(1+R_t^{1/2}W(z)R_t^{1/2}\Big)^{-1} R_t^{1/2}.  \label{eq: RS(Z)} 
\end{equation}
Thus we get
\begin{align}
 (S(z)^2+t^2)^{-1}-(T^2+t^2)^{-1}
 & = -(S(z)^2+t^2)^{-1}W(z)(T^2+t^2)^{-1} \\
 & = -R_t^{1/2}\Big(1+R_t^{1/2}W(z)R_t^{1/2}\Big)^{-1} R_t^{1/2} W(z) R_t.
\end{align}
This completes the proof.
\end{proof}

\begin{Lemma}\label{main1 lemma4}
Assume conditions (C1)--(C6). Let $z\in U$ be arbitrary.
Then the range of $W'(z)T^{-1}$ is contained in the domain of $T^{-1}$, 
the operator $T^{-1}W'(z)T^{-1}$ is bounded, and its closure is of trace class with
\begin{equation}
   \frac{d}{dz}\,\ovl{T^{-1}W(z)T^{-1}}=\ovl{T^{-1}W'(z)T^{-1}}.
\end{equation}
In particular, the map
\begin{equation}
   U\to\mathcal{S}_1(\sH ),\qquad z\mapsto\ovl{T^{-1}W'(z)T^{-1}}
\end{equation}
is continuous in the trace norm.
\end{Lemma}

\begin{proof}
We set 
\begin{equation}
   F(z) \coloneqq  \ovl{T^{-1}W(z)T^{-1}},\qquad z\in U.
\end{equation}
For any $u,v\in\dom(T^{-1})$, we have
\begin{equation}
   \langle u, F(z)v\rangle = \langle T^{-1}u, W(z)T^{-1}v\rangle.
\end{equation}
  Differentiating both sides with respect to $z$, we get
\begin{equation}
   \langle u,F'(z)v\rangle = \langle T^{-1}u,W'(z)T^{-1}v\rangle.
\end{equation}
This implies that $W'(z)T^{-1}v\in\dom(T^{-1})$ and $T^{-1}W'(z)T^{-1}v = F'(z)v$.
Since $v\in\dom(T^{-1})$ is arbitrary and $F'(z) \in \cB(\sH)$, 
the operator $T^{-1}W'(z)T^{-1}$ is bounded, and its closure is equal to $F'(z)$.
This finishes the proof.
\end{proof}

\begin{proof}[Proof of Theorem \ref{main1}]
For $t>0$ and $z\in U$, we set 
\begin{equation}
 f(z,t) \coloneqq  (S(z)^2+t^2)^{-1} - (T^2+t^2)^{-1}.
\end{equation}
It then follows from \eqref{int_rep_S(z)} and the resolvent identity that
\begin{equation}
  (S(z)-T)u = -\frac{2}{\pi} \int_0^\infty f(z,t) u \, t^2 dt   \label{eq:int_repr_S-T}
\end{equation}
for all $u\in \dom(T^2)$.

Let $K \subset U$ be an arbitrary non-empty compact subset. 
By condition (C5) and Lemma \ref{main1 lemma4}, the maps $z \mapsto W(z)$ and
 $z \mapsto \ovl{T^{-1}W(z)T^{-1}}$ are $\cS_1(\sH)$-valued holomorphic functions, 
and hence continuous, on $U$. 
Thus we can define
\begin{equation}
  M_1 \coloneqq  \sup_{z \in K} \|W(z)\|_1 < \infty \quad \text{and} \quad 
  M_2 \coloneqq  \sup_{z \in K} \|\ovl{T^{-1}W(z)T^{-1}}\|_1 < \infty.
\end{equation}
For notational simplicity, we set 
\begin{equation}
  \til{W}_t(z) \coloneqq  R_t^{1/2} W(z) R_t^{1/2}.   \label{eq:W_tilde}
\end{equation}
If $t \geq 1$, it follows from Lemma \ref{main1 lemma2} and Lemma \ref{main1 lemma3} that
\begin{align}
  \| f(z,t) \|_1
  &\leq \|R_t^{1/2}\| \big\|\big(1+ \til{W}_t(z)\big)^{-1}\big\| \|R_t^{1/2}\| \|W(z)\|_1 \|R_t\| \\
  &\leq \frac{1}{t} \cdot \frac{1}{\delta} \cdot \frac{1}{t} \cdot M_1 \cdot \frac{1}{t^2} = \frac{M_1}{\delta t^4}
\end{align}
for all $z \in K$. Similarly, if $0 < t \leq 1$, we have
\begin{align}
  \| f(z,t) \|_1
  & = \big\|R_t^{1/2}\big(1+\til{W}_t(z) \big)^{-1}TR_t^{1/2} \cdot \ovl{T^{-1}W(z)T^{-1}} \cdot TR_t\big\|_1 \\
  &\leq \|R_t^{1/2}\| \big\|\big(1+\til{W}_t(z) \big)^{-1}\big\| \|TR_t^{1/2}\| \big\|\ovl{T^{-1}W(z)T^{-1}}\big\|_1 \|TR_t\| \\
  &\leq \frac{1}{t} \cdot \frac{1}{\delta} \cdot 1 \cdot M_2 \cdot \frac{1}{t} = \frac{M_2}{\delta t^2}
\end{align}
for all $z \in K$. Combining these two cases, we obtain the uniform bound over $K$:
\begin{equation}
  \sup_{z \in K} \int_0^\infty \| f(z,t) \|_1 \, t^2 dt 
  \leq \int_1^\infty \frac{M_1}{\delta t^2} \, dt + \int_0^1 \frac{M_2}{\delta} \, dt 
   = \frac{M_1+M_2}{\delta} < \infty. \label{eq:uniform_bound}
\end{equation}
This integrability implies that for each $z \in U$, the $\cS_1(\sH)$-valued Bochner integral
\begin{equation}
  A(z) \coloneqq  -\frac{2}{\pi} \int_0^\infty f(z,t) \, t^2 dt
\end{equation}
converges absolutely in $\cS_1(\sH)$. 
In view of \eqref{eq:int_repr_S-T}, the operator $S(z)-T$ coincides with $A(z)$ on the dense domain $\dom(T^2)$.
This proves that the restriction $(S(z)-T)|_{\dom(T^2)}$ is bounded, and its extension 
$\ovl{(S(z)-T)|_{\dom(T^2)}} = A(z)$ belongs to $\cS_1(\sH)$.

To establish the holomorphy of the map $z \mapsto A(z)$, we employ Mattner's theorem (Theorem \ref{thm:comp-diff}). 
By Lemma~\ref{main1 lemma2}, Lemma~\ref{main1 lemma3} and condition (C2), 
the integrand 
\begin{equation}
 f(z,t) = - R_t^{1/2}\big(1+R_t^{1/2} W(z) R_t^{1/2}\big)^{-1} R_t^{1/2} W(z) R_t
\end{equation}
is an $\cS_1(\sH)$-valued holomorphic function in $z \in U$ for each $t > 0$.

Let $\varphi$ be an arbitrary bounded linear functional on $\cS_1(\sH)$. 
Then, the function $z \mapsto \varphi(f(z,t))$ is holomorphic on $U$. 
Furthermore, the uniform bound \eqref{eq:uniform_bound} yields
\begin{equation}
  \sup_{z \in K} \int_0^\infty |\varphi(f(z,t))| \, t^2 dt 
  \leq \|\varphi\| \sup_{z \in K} \int_0^\infty \|f(z,t)\|_1 \, t^2 dt < \infty.
\end{equation}
Applying Theorem \ref{thm:comp-diff} to this function,
we conclude that the integrated function
\begin{equation}
  \varphi(A(z)) = -\frac{2}{\pi} \int_0^\infty \varphi(f(z,t)) \, t^2 dt
\end{equation}
is holomorphic on $U$. 
Since $\varphi \in \cS_1(\sH)^*$ was arbitrary, Dunford's theorem ensures that the map 
 $z \mapsto A(z)$ is an $\cS_1(\sH)$-valued holomorphic function on $U$.

Finally, since the trace $\tr : \cS_1(\sH) \to \CC$ is a bounded linear functional,
 the holomorphy of the map $z \mapsto \tr(A(z))$ follows immediately.
This completes the proof.
\end{proof}

\begin{Corollary}\label{main1 cor1}
 Assume conditions (C1)--(C6).
Let $U_1$ be the set of all $z\in U$ satisfying $\|\ovl{T^{-1}W(z)T^{-1}}\|<1$.
Then $U_1$ is an open subset of $\CC$, and we have
\begin{equation}
  \tr\Big( \ovl{[S(z)-T]|_{\dom(T^2)}} \Big)
  = \frac{2}{\pi} \sum_{m=1}^\infty (-1)^{m-1} \int_0^\infty \tr \big(R_t(W(z)R_t)^m\big) \, t^2 dt,
   \qquad \forall z\in U_1.
\end{equation}
\end{Corollary}

\begin{proof}
Since the function $\big\|\ovl{T^{-1}W(z)T^{-1}}\big\|$ is continuous in $z\in U$, 
the set $U_1$ is open in $\CC$.
 Next, combining what we showed in the proof of Theorem \ref{main1} 
with Lemma \ref{main1 lemma3}, we obtain
\begin{align}
  \tr\Big( \ovl{ (S(z)-T)|_{\dom(T^2)} } \Big)
  & = -\frac{2}{\pi}\int_0^\infty \tr\left[(S(z)^2+t^2)^{-1}-(T^2+t^2)^{-1}\right] \, t^2dt \\
  & = \frac{2}{\pi}\int_0^\infty \tr
      \Big[ R_t^{1/2}\big(1+\til{W}_t(z) \big)^{-1}R_t^{1/2}W(z)R_t \Big]
      \, t^2dt. \label{eq:integ_repr_tr}
\end{align}
For any $z\in U_1$ and $t>0$, it holds that
\begin{equation}
   \|\til{W}_t(z) \| \leq \big\| \ovl{T^{-1}W(z)T^{-1}} \big\| <1,  \label{main1cor1eq}
\end{equation}
and thus we get
\begin{equation}
 \big(1+ \til{W}_t(z) \big)^{-1} = \sum_{m=0}^\infty\big(-\til{W}_t(z) \big)^m
\end{equation}
in the operator norm. Since $W(z)$ is of trace class, we have
\begin{align}
& \tr\Big(\ovl{[S(z)-T]|_{\dom(T^2)}}\Big)  \\
& = \frac{2}{\pi} \int_0^\infty \sum_{m=0}^\infty(-1)^{m} \tr
    \left[ R_t^{1/2} \Big(R_t^{1/2}W(z)R_t^{1/2}\Big)^{m} R_t^{1/2} W(z) R_t\right] \, t^2 dt  \\ 
& = \frac{2}{\pi}\int_0^\infty\sum_{m=0}^\infty(-1)^{m} \tr\left[R_t\Big(W(z)R_t\Big)^{m+1}\right] \, t^2 dt \\
& = \frac{2}{\pi}\int_0^\infty\sum_{m=1}^\infty(-1)^{m-1} \tr\Big(R_t(W(z)R_t)^m\Big) \, t^2 dt.
\end{align}

To complete the proof, it suffices to interchange the integral and the infinite sum above.
For this, let us estimate the integrand. Since
\begin{equation}
\norm{\til{W}_t(z)}_1 \leq  \min \{ t^{-2} \norm{W(z)}_1, \norm{\ovl{T^{-1}W(z)T^{-1}}}_1 \},  
\end{equation}
we have
\begin{align}
& \left|\tr\big(R_t(W(z)R_t)^m\big)\right|
 = \big|\tr \big( R_t^{1/2} \big(\til{W}_t(z) \big)^m R_t^{1/2} \big)\big|  \\
& \leq \norm{R_t^{1/2}} \norm{\til{W}_t(z)}_1 \norm{\til{W}_t(z)}^{m-1} \norm{R_t^{1/2}} \\
& \leq \big\| \ovl{T^{-1}W(z)T^{-1}} \big\|^{m-1} 
      \min \{ t^{-2}\norm{W(z)}_1, \norm{ \ovl{T^{-1}W(z)T^{-1}} }_1 \} \, t^{-2},
\end{align}
for $t>0$ and $m=1,2,\cdots$.
Thus, by \eqref{main1cor1eq}, we have
\begin{equation}
   \int_0^\infty\sum_{m=1}^\infty \left|\tr\big( R_t (W(z)R_t)^m \big) \right| \, t^2 dt < \infty,
\end{equation}
which allows us to interchange the sum and integral.
This completes the proof.
\end{proof}

\section{Holomorphy of the Ground State}\label{sect:GS}

The proof of the holomorphy of the ground state is more involved than that of the ground state energy.

Consider the unnormalized ground state \eqref{phitil} of $H$.
Here, $K_1$ is defined by \eqref{def:Ks} in Appendix \ref{AppenC} 
using $X$ and $Y$ as defined in \eqref{defXY}.
We consider a situation where the Hamiltonian depends on a coupling constant $\alpha \in \RR$.
In this case, $K_1$ depends on $\alpha$, so we denote it by $K_1(\alpha)$.
First, we extend $K_1(\alpha)$ analytically to an operator-valued holomorphic function
$K_1(z)$ defined on an open subset of $\CC$.

Next, we show that $K_1(z)$ satisfies the conditions of Theorem \ref{gs:Holo} in Appendix \ref{AppenC}.
 This establishes the holomorphy of the ground state.

In what follows, similarly to Section \ref{sect:GSE}, we will establish holomorphy in a general 
framework without explicitly introducing $\alpha$. 
The application to specific models will be discussed in Section \ref{sect:example}.

\subsection{\texorpdfstring{Definition of $K_1(z)$}{Definition of K1(z)}}
Following Proposition \ref{prop:K1_K4}, we introduce an operator $K_4(z)$:
\begin{Theorem}\label{thm:K4_1}
 Assume conditions (C1)--(C6). 
For each $z \in U$, the operator
\begin{equation}
   K_4(z) \coloneqq  \frac{1}{2} \ovl{T^{-1/2}(S(z) - T)T^{-1/2}},
\end{equation}
is a bounded operator on $\sH$
and is of trace class.
In particular, $K_4(z)$ is a Hilbert--Schmidt operator.
\end{Theorem}

\begin{proof}
Recall that $S(z)$ admits the integral representation \eqref{int_rep_S(z)} under the assumed conditions. 
We fix $z \in U$ and, for simplicity, omit the argument $z$ by writing $S = S(z)$, $W = W(z)$, 
and $\til{W}_t = \til{W}_t(z) \coloneqq  R_t^{1/2}W(z)R_t^{1/2}$. 
As shown in the proof of Theorem~\ref{main1}, the resolvent formula yields
\begin{align}
(S - T)u 
& = \frac{2}{\pi} \int_0^\infty R_t^{1/2} \til{W}_t (1 + \til{W}_t)^{-1} R_t^{1/2} u \, t^2 dt \\
& = \frac{2}{\pi} \int_0^\infty R_t^{1/2} 
\left\{ \til{W}_t - \til{W}_t (1 + \til{W}_t)^{-1} \til{W}_t \right\} R_t^{1/2} u \, t^2 dt,
\end{align}
for all $u \in \dom(T^2)$.

Consequently, for any $u\in \dom(T^{3/2})\cap \dom(T^{-1/2})$ and $v\in \dom(T^{-1/2})$,
a direct computation gives
\begin{align}
&  \inner{T^{-1/2}v}{(S-T)T^{-1/2}u} \\
& = \frac{2}{\pi} \int_0^\infty \inner{T^{-1/2}v}{R_t^{1/2} \til{W}_t  R_t^{1/2} T^{-1/2} u} \, t^2 dt \\
& \quad  -\frac{2}{\pi} \int_0^\infty 
  \inner{T^{-1/2}v}{R_t^{1/2} \til{W}_t( 1 + \til{W}_t)^{-1}\til{W}_t  R_t^{1/2} T^{-1/2} u} \, t^2 dt \\
& = \frac{2}{\pi} \int_0^\infty \inner{v}{ T^{1/2} R_t \ovl{T^{-1} W T^{-1}} T^{1/2} R_tu} \, t^2 dt \\
& \quad  -\frac{2}{\pi} \int_0^\infty 
  \inner{v}{T^{1/2}R_t  \ovl{T^{-1}WT^{-1}}  Q_t  \ovl{T^{-1}WT^{-1}}  T^{1/2}R_t u} \, t^2 dt,
\end{align}
where
\begin{equation}
  Q_t \coloneqq  TR_t^{1/2} (1 + \til{W}_t)^{-1} TR_t^{1/2}.
\end{equation}
We now define the operators
\begin{align}
  I_1 & \coloneqq  \frac{2}{\pi} \int_0^\infty  T^{1/2} R_t \ovl{T^{-1} W T^{-1}} T^{1/2} R_t  \, t^2 dt, \label{def:I_1}\\
  I_2 & \coloneqq  \frac{2}{\pi} \int_0^\infty  T^{1/2} R_t  \ovl{T^{-1}WT^{-1}}  Q_t  \ovl{T^{-1}WT^{-1}}  T^{1/2}R_t \, t^2 dt. \label{def:I_2}
\end{align}
By Theorem~\ref{thm:trace_est1}, $ I_1 $ is a trace class operator defined on the whole space $\sH$.
For $I_2$, combining Lemma~\ref{main1 lemma1} with the inequality $\norm{TR_t^{1/2}} \leq 1$ yields
\begin{equation}
   \| Q_t \| \leq \delta(z)^{-1}, \qquad t > 0.
\end{equation}
This bound, together with condition (C4), allows us to apply Theorem \ref{thm:trace_est3}, 
which implies that $ I_2 $ is also a trace class operator defined on $\sH$. 

Therefore, for all $ u \in \dom(T^{3/2}) \cap \dom(T^{-1/2}) $ and $ v \in \dom(T^{-1/2}) $, we obtain
\begin{equation}
   \inner{T^{-1/2}v}{(S - T)T^{-1/2}u} = \inner{v}{(I_1 - I_2)u}.
\end{equation}
This identity shows that $(S - T)T^{-1/2}u \in \dom(T^{-1/2})$ and $ T^{-1/2}(S - T)T^{-1/2}u = (I_1 - I_2)u$
on the dense domain.
Since $I_1 - I_2$ is a trace class operator defined on the whole space $\sH$, 
$T^{-1/2}(S - T)T^{-1/2}$ is a densely defined bounded operator, and its closure $\ovl{T^{-1/2}(S - T)T^{-1/2}} = I_1 - I_2$ is of trace class.
\end{proof}

The main result of this subsection is as follows.
\begin{Theorem}\label{hs of K1}
Assume conditions (C1)--(C6). 
Then, for $z \in U$, the operator $1 + K_4(z)$ has a bounded inverse, and 
$K_1(z) \coloneqq  K_4(z)(1 + K_4(z))^{-1}$ is a trace class operator. 
In particular, $K_1(z)$ is Hilbert--Schmidt.
\end{Theorem}

\begin{proof}
Let $z \in U$. Note that $K_4(z)$ can be written as follows:
\begin{equation}
  K_4(z) = \frac{1}{2}\ovl{T^{-1/2}S(z)T^{-1/2}} - \frac{1}{2}.
\end{equation}
Since $S(z)$ is accretive (as mentioned just above \eqref{int_rep_S(z)}),
$\ovl{T^{-1/2}S(z)T^{-1/2}}$ is m-accretive. 
In particular, $-1$ is in the resolvent set of $\ovl{T^{-1/2}S(z)T^{-1/2}}$.
Therefore, $K_4(z) + 1$ has a bounded inverse.
Since $K_4(z)$ is trace class by Theorem \ref{thm:K4_1}, it follows that $K_1(z)$ is also trace class.
\end{proof}

\subsection{\texorpdfstring{Holomorphy of $K_1(z)$}{Holomorphy of K1(z)}}

\begin{Theorem}\label{thm:holo_K1}
Assume conditions (C1)--(C6). Then the map
\begin{equation}
  U \to \cS_1(\sH), \qquad z \mapsto K_1(z)
\end{equation}
is holomorphic.
\end{Theorem}
\begin{proof}
Let $I_1(z)$ and $I_2(z)$ be the trace class operators defined in \eqref{def:I_1} and \eqref{def:I_2}, 
where we make the $z$-dependence explicit.
Recall from Theorem~\ref{thm:K4_1} that $2K_4(z) = I_1(z) - I_2(z)$.

We first prove the holomorphy of $I_1(z)$. 
By Lemma~\ref{main1 lemma4}, the integrand of $I_1(z)$ is an $\cS_1(\sH)$-valued holomorphic function in $z \in U$ for each $t > 0$. 
Let $\varphi$ be an arbitrary bounded linear functional on $\cS_1(\sH)$. 
Then, for any non-empty compact subset $K \subset U$, we have
\begin{align}
  &\sup_{z \in K} \int_0^\infty \big| \varphi\big(T^{1/2} R_t \ovl{T^{-1} W(z) T^{-1}} T^{1/2} R_t\big) \big| \, t^2 dt \\
  &\quad \leq \|\varphi\| \sup_{z \in K} \int_0^\infty \big\| T^{1/2} R_t \ovl{T^{-1} W(z) T^{-1}} T^{1/2} R_t \big\|_1 \, t^2 dt.
\end{align}
Applying Theorem~\ref{thm:trace_est1} and condition (C5), we obtain a uniform bound
\begin{equation}
  \sup_{z \in K} \int_0^\infty \big| \varphi\big(T^{1/2} R_t \ovl{T^{-1} W(z) T^{-1}} T^{1/2} R_t\big) \big| \, t^2 dt 
  \leq \frac{\pi}{4} \|\varphi\| \sup_{z \in K} \big\| \ovl{T^{-1} W(z) T^{-1}} \big\|_1 < \infty.
\end{equation}
Thus, Theorem~\ref{thm:comp-diff} implies that $\varphi(I_1(z))$ is holomorphic on $U$. 
Since $\varphi \in \cS_1(\sH)^*$ was arbitrary, Dunford's theorem ensures that $I_1(z)$ is 
an $\cS_1(\sH)$-valued holomorphic function on $U$.

Similarly, using Lemma~\ref{main1 lemma2} and Theorem~\ref{thm:trace_est3}, one can establish the holomorphy of $I_2(z)$. 
Consequently, $K_4(z)$ is an $\cS_1(\sH)$-valued holomorphic function on $U$. 
Since $z \mapsto (1 + K_4(z))^{-1}$ is $\cB(\sH)$-holomorphic, $K_1(z) = K_4(z)(1 + K_4(z))^{-1}$ is also holomorphic in $\cS_1(\sH)$ on $U$, which completes the proof.
\end{proof}

\subsection{\texorpdfstring{$\norm{K_1(z)}<1$}{||K1(z)|| < 1}}
\begin{Lemma}\label{Lem:K1norm}
Assume conditions (C1)--(C6). 
Then, $\norm{K_1(z)} < 1$ for all $z\in U$.
\end{Lemma}

\begin{proof}
As shown in the proof of Proposition \ref{prop:m-acc}, we have
\begin{equation}
  \delta(z) \inner{u}{T^2 u} \leq  \Re\! \inner{u}{(T^2+W(z))u}, \qquad u\in\dom(T^2).
\end{equation}
From this, by applying Theorem \ref{thm:g-Heinz1} with $ A = \delta(z) T^2, B_0 = T^2 + \Re W(z)$,
 and $B_1 = \Im W(z)$, we obtain
\begin{equation}
  \delta(z)^{1/2} \inner{u}{Tu}
  \leq \Re \inner{u}{(T^2 + W(z))^{1/2}u},
  \qquad u \in \dom(T^2).  \label{ineq:240119}
\end{equation}
For any $u\in\sH$, we set $u_n\coloneqq  E_T((n^{-1},n))u~(n=1,2,\cdots)$, where 
$E_T$ denotes the spectral measure of $T$.
Then, $u_n\to u~(n\to\infty)$ and
\begin{equation}
  2\inner{u}{K_4(z)u} 
  = 2\lim_{n\to\infty} \inner{u_n}{K_4(z)u_n}
  = \lim_{n\to\infty} \inner{T^{-1/2}u_n}{S(z)T^{-1/2} u_n} - \norm{u}^2
\end{equation}
hold.
Applying \eqref{ineq:240119} to $T^{-1/2}u_n$ yields
\begin{equation}
  2\Re \inner{u_n}{K_4(z)u_n} \geq  \big(\delta(z)^{1/2}- 1\big)\norm{u_n}^2.
\end{equation}
Taking $n\to\infty$, we obtain
\begin{equation}
  2 \Re \inner{u}{K_4(z)u} \geq 
    \big(\delta(z)^{1/2}- 1\big)\norm{u}^2.
\end{equation}
Therefore, we have
\begin{align}
  \norm{(1+K_4(z))u}^2 
& = \norm{u}^2 + 2\Re \inner{u}{K_4(z)u} + \norm{K_4(z)u}^2 \\
& \geq \norm{K_4(z)u}^2 + \delta(z)^{1/2} \norm{u}^2.
\end{align}
Since $1+K_4(z)$ is surjective, for any $v\in\sH$, we have
\begin{equation}
  \norm{K_4(z)(1+K_4(z))^{-1}v}^2 \leq \norm{v}^2 - \delta(z)^{1/2} \norm{(1+K_4(z))^{-1}v}^2.
\end{equation}
By noting that 
\begin{equation}
 1 \leq \norm{1+K_4(z)}^2 \norm{(1+K_4(z))^{-1}v}^2, \qquad (\norm{v}=1),
\end{equation}
we have
\begin{equation}
  \norm{K_1(z)}^2 
  = \norm{K_4(z)(1+K_4(z))^{-1}}^2
  \leq 1- \delta(z)^{1/2} \norm{1+K_4(z)}^{-2} <1.
\end{equation}
\end{proof}

\subsection{Analytic Continuation of the Ground State}
The ground state \eqref{phitil} is extended via the substitution $ K_1 \to K_1(z) $ as follows:
\begin{align}
  \til{\Phi}_\mathrm{g}(z)
  & \coloneqq  e^{-\frac{1}{2}\Delta^*(K_1(z))} \Omega\\
  & = \sum_{n=0}^\infty \left(-\frac{1}{2}\right)^n \frac{1}{n!} (\Delta^*(K_1(z)))^n \Omega.   \label{extGS}
\end{align}
The right-hand side converges in the Fock space $\Fb(\sH)$ and is holomorphic in $z \in U$.
\begin{Theorem} \label{main2}
  Assume conditions (C1)--(C6). 
Then, for any $z \in U$, the right-hand side of \eqref{extGS} converges absolutely 
in the norm of the Fock space $\Fb(\sH)$,
  and the map
 \begin{equation}
   U \to \Fb(\sH), \qquad z\mapsto \Phigt(z)
 \end{equation}
 is holomorphic.
\end{Theorem}

\begin{proof}
From Theorem \ref{hs of K1}, $K_1(z)$ is Hilbert--Schmidt, and from Lemma \ref{Lem:K1norm}, 
we have $\norm{K_1(z)} < 1$. Therefore, by Theorem \ref{thm:squeezed-state}, the series \eqref{extGS}
converges absolutely. By Theorem \ref{gs:Holo}, $\Phigt(z)$ is holomorphic.
\end{proof}

\begin{Corollary}
Assume conditions (C1)--(C6).
If $0\in U$ and $K_1(0)=0$, then $\Phigt(\alpha)$ has a convergent Maclaurin series:
\begin{equation}
  \Phigt(\alpha) 
 = \Omega + \sum_{m=1}^\infty \frac{\alpha^m}{m!} 
   \bigg(\frac{d^m \Phigt(\alpha)}{d\alpha^m}\bigg|_{\alpha=0}\bigg),
\qquad (|\alpha|< \mathrm{dist}(0,U^\mathrm{c})).
\end{equation}
\end{Corollary}

\section{Examples}\label{sect:example}
\subsection{The single pair interaction model}
Let $T$ be an injective non-negative self-adjoint operator acting on a separable complex Hilbert space $\sH$.
The Hamiltonian for the single pair interaction model is given by
\begin{equation}
  H(\lambda) \coloneqq  \dGb(T) + \frac{\lambda}{2} \PhiS(g)^2,
\end{equation}
where $g$ is a non-zero vector in $\dom(T^{1/2})\cap \dom(T^{-1/2})$ and $\lambda\in\RR$ is a coupling constant.
Note that $H(\lambda)$ acts in $\Fb(\sH)$.
This model is a special case of the one described in Section \ref{sec:Definition_of_the_Model}.
For the diagonalization and spectral analysis of this model, see \cite{AF, Gamet2026} and \cite[Section 7.1]{MR4213757}.
We investigate the holomorphy of the ground state and the ground state energy of this model.
Furthermore, we can precisely determine the radii of convergence of the Maclaurin series for the ground state and its energy.

Since the following discussion is based on \cite[Section 7.1]{MR4213757}, we review the necessary setup.
We fix a conjugation $J$ satisfying $JTJ = T$ and $Jg = g$ (such a $J$ exists by \cite[Lemma 7.1]{MR4213757}).  
The operator $S(\lambda)$ is defined by
\begin{equation}  
  S(\lambda) \coloneqq  \big( T^2 + \lambda \ketbra{T^{1/2}g}{T^{1/2}g} \big)^{1/2}.  
\end{equation}

\begin{Theorem}[{\cite[Theorem 7.2]{MR4213757}}]\label{diagonalization of SPIM}
Suppose $1+\lambda\|T^{-1/2}g\|^2>0$. 
Then, $H(\lambda)$ is self-adjoint, and essentially self-adjoint on any core of $\dGb(T)$.
Furthermore, the unitary operator $\sU=\sU(\lambda)$ on $\Fb(\sH)$ that satisfies equation \eqref{defU*} yields
\begin{equation}
  \sU H(\lambda) \sU^* = \dGb( S(\lambda) ) + E(\lambda), \qquad 
  E(\lambda) = \frac{1}{2}\tr\left(\, \ovl{S(\lambda)-T}\, \right).  
\end{equation}
\end{Theorem}

To discuss holomorphy, we define
\begin{equation}
   W(z) \coloneqq  z |T^{1/2}g\rangle \langle T^{1/2}g|,    \qquad z\in\CC.
\end{equation}
The pair $(T,W)$ satisfies conditions (C1)--(C5) with $U=\CC$.
Since the condition (C6) does not hold with $U=\CC$, we set
\begin{equation}
   U \coloneqq  \{ z\in \CC \mid \Re(z)>-\|T^{-1/2}g\|^{-2} \}. \label{def of Uspm}
\end{equation}

\begin{Proposition}
Let $z$ be a complex number. 
Then $\Re( 1 +\ovl{T^{-1}W(z)T^{-1}}) >0$ holds if and only if $z\in U$.
That is, the pair $(T,W)$ satisfies conditions (C1)--(C6) with $U$ defined by \eqref{def of Uspm}.
In particular, $S(z)$ can be defined for all $z \in U$.
\end{Proposition}
\begin{proof}
Note that the condition $\Re( 1 +\ovl{T^{-1}W(z)T^{-1}}) >0$ is equivalent to 
the condition that the operator $1+\ovl{T^{-1}W(z)T^{-1}}-\delta$ is accretive for 
some positive constant $\delta>0$.
 
Suppose that $1+\ovl{T^{-1}W(z)T^{-1}}-\delta$ is accretive for some $\delta>0$.
Then for any $f\in\sH$, it holds that
\begin{equation}
  \delta \|f\|^2
  \leq \Re \biginner{f}{\big(1+\ovl{T^{-1}W(z)T^{-1}}\big)f}
  = \|f\|^2 + \Re(z)|\langle T^{-1/2}g,f\rangle|^2.
\end{equation}
By taking the supremum over all unit vectors $f\in\sH$, we get
\begin{equation}
   1+\Re(z)\|T^{-1/2}g\|^2\geq \delta>0,
\end{equation}
and hence $z\in U$.
 Conversely, let $z$ be an arbitrary complex number in $U$. 
We set $\delta_0 \coloneqq  1+\Re(z)\|T^{-1/2}g\|^{2}>0$ and $\delta \coloneqq  \min\{\delta_0,1\}\in (0,1]$.
Take any $f\in\sH$, and write it as $f=\alpha T^{-1/2}g+h$ for some $\alpha\in\CC$ and $h\in(\CC T^{-1/2}g)^\perp$.
 Then it follows that
\begin{align}
 & \Re \biginner{f}{\big(1+\ovl{T^{-1}W(z)T^{-1}}\big)f}
   = \|f\|^2+\Re(z)|\langle T^{-1/2}g,f\rangle|^2\\
 & =|\alpha|^2\|T^{-1/2}g\|^2+\|h\|^2+\Re(z)|\alpha|^2\|T^{-1/2}g\|^4\\
 & =|\alpha|^2\|T^{-1/2}g\|^2(1+\Re(z)\|T^{-1/2}g\|^2)+\|h\|^2\\
 & \geq  |\alpha|^2\|T^{-1/2}g\|^2\delta+\|h\|^2\delta=\delta\|f\|^2,
\end{align}
whence $1+\ovl{T^{-1}W(z)T^{-1}}-\delta$ is accretive.
\end{proof}

Based on the above proposition and the discussions in Section \ref{sect:GSE}, 
the ground state energy $E(\lambda)$ can be extended to a holomorphic function on $U$. 
Moreover, by the discussions in Section \ref{sect:GS}, $K_1(\lambda)$ can be 
analytically continued to a holomorphic function on $U$, and therefore, 
the ground state $\Phigt(\lambda)$ can also be extended to a holomorphic function on $U$.
Note that $\Phigt(z)$ is defined by \eqref{extGS} and 
\begin{equation}
  E(z) \coloneqq  \frac{1}{2} \tr \ovl{(S(z)-T)|_{\dom(T^2)}}.  \label{single pair model eq}
\end{equation}

The holomorphy of the ground state and the ground state energy can be stated as follows.
\begin{Theorem}\label{thm:SPIM_GS_conv_rad}
The extended ground state $\Phigt(z)$ and ground state energy $E(z)$ are
holomorphic on $U$ defined by \eqref{def of Uspm}.
Moreover, the radius of convergence of the Maclaurin series of $E(\lambda)$ 
is equal to $\norm{T^{-1/2}g}^{-2}$.
\end{Theorem}

\begin{proof}
By Theorems \ref{main1} and \ref{main2}, we conclude that $E(z)$ and $\Phigt(z)$ are holomorphic on $U$.
Finally, let us compute the radius of convergence of the Maclaurin series of \eqref{single pair model eq}.
Since the function \eqref{single pair model eq} is holomorphic on $\{z\in\CC \mid \Re(z)>-\|T^{-1/2}g\|^{-2}\}$, 
the radius is greater than or equal to $\|T^{-1/2}g\|^{-2}$.

Suppose that the radius is strictly greater than $\|T^{-1/2}g\|^{-2}$.
Take a sufficiently small $\vep>0$ such that $\|T^{-1/2}g\|^{-2}(1+\vep)$ is smaller than the radius.
By Corollary \ref{main1 cor1}, we have
\begin{align}
& \frac{1}{2} \tr\Big(\ovl{[S(z)-T]|_{\dom(T^2)}}\Big)
  = \frac{1}{\pi} \sum_{m=1}^\infty(-1)^{m-1}\int_0^\infty \tr \Big(R_t(W(z)R_t)^m\Big) \, t^2 dt \\
& = \sum_{m=1}^\infty z^m\cdot\frac{(-1)^{m-1}}{\pi}\int_0^\infty\|R_t^{1/2}T^{1/2}g\|^{2m-2}\|R_tT^{1/2}g\|^2 \, t^2 dt \label{spm taylor ex}
\end{align}
for all sufficiently small $z$.
Since this is the Maclaurin series of the function \eqref{single pair model eq}, 
the right-hand side of \eqref{spm taylor ex} absolutely converges for all $z$ within the radius of convergence.
In particular, it converges absolutely when $z=\|T^{-1/2}g\|^{-2}(1+\vep)$.

Since $\|R_t^{1/2}T^{1/2}g\|$ converges to $\|T^{-1/2}g\|$ as $t\to+0$, there exists a constant $t_0>0$ such that
\begin{equation}
   \frac{(1+\vep)\|R_{t_0}^{1/2}T^{1/2}g\|^2}{\|T^{-1/2}g\|^2}>1.
\end{equation}
Then it holds that
\begin{align}
 \infty 
& > \sum_{m=1}^\infty \frac{(1+\vep)^m}{\|T^{-1/2}g\|^{2m}} \int_0^\infty \|R_t^{1/2}T^{1/2}g\|^{2m-2} \|R_tT^{1/2}g\|^2 \, t^2dt\\
&\geq \sum_{m=1}^\infty \frac{(1+\vep)^m}{\|T^{-1/2}g\|^{2m}}\int_0^{t_0}\|R_t^{1/2}T^{1/2}g\|^{2m-2}\|R_tT^{1/2}g\|^2\, t^2dt   \\
&\geq \sum_{m=1}^\infty \frac{(1+\vep)^m}{\|T^{-1/2}g\|^{2m}}\int_0^{t_0}\|R_{t_0}^{1/2}T^{1/2}g\|^{2m-2}\|R_tT^{1/2}g\|^2\, t^2dt\\
& = \frac{1}{\|R_{t_0}^{1/2}T^{1/2}g\|^2}\int_0^{t_0}\|R_tT^{1/2}g\|^2\, t^2dt 
   \cdot \sum_{m=1}^\infty \left(\frac{(1+\vep)\|R_{t_0}^{1/2}T^{1/2}g\|^2}{\|T^{-1/2}g\|^{2}}\right)^m\\
&=\infty,
\end{align}
which is a contradiction.
Hence, the radius of convergence is precisely $\|T^{-1/2}g\|^{-2}$.
\end{proof}

\begin{Theorem}\label{thm:conv_rad_gs_single}
The radius of convergence of the Maclaurin series of the ground state $\Phigt(\lambda)$ is $\|T^{-1/2}g\|^{-2}$.
\end{Theorem}

\begin{proof}
Suppose that the radius of convergence of the Maclaurin series for 
$\Phigt(\lambda)$ is $r>\|T^{-1/2}g\|^{-2}$.
For any $\lambda > -\|T^{-1/2}g\|^{-2}$, we have
\begin{equation}
  \langle\Omega,H(\lambda)\Phigt(\lambda)\rangle
 = E(\lambda)\langle\Omega,\Phigt(\lambda)\rangle 
 = E(\lambda),
\end{equation}
by our normalization. On the other hand,
\begin{equation}
 \langle\Omega,H(\lambda)\Phigt(\lambda)\rangle
 = \langle H(\lambda)\Omega,\Phigt(\lambda)\rangle 
 = \frac{\lambda}{2}\langle\Phi_\mathrm{S}(g)^2\Omega,\Phigt(\lambda)\rangle.	
\end{equation}
The right-hand side extends holomorphically to $\{z \in \CC \mid |z|<r\}$. 
Hence, so does $E(\lambda)$, which contradicts Theorem \ref{thm:SPIM_GS_conv_rad}.
\end{proof}

\begin{Remark}
The relationship between the region $U$ and the radius of convergence is illustrated in Figure \ref{fig:region_U_spm}.
\end{Remark}

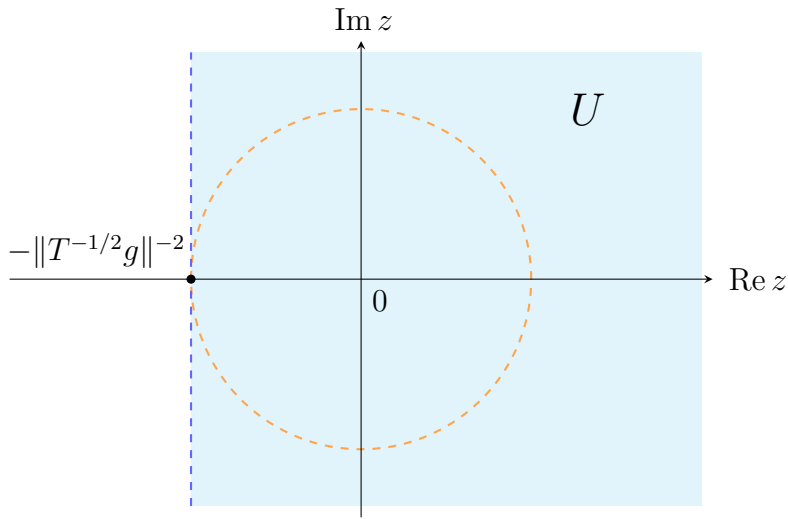
\begin{figure}[htbp]
    \centering
    \begin{tikzpicture}[>=stealth, scale=1.5]
        \def\invNorm{1.5}  
        \def\Xrange{3.0}   
        \def\Yrange{2.0}  

        \fill[cyan!10] (-\invNorm, -\Yrange) rectangle (\Xrange, \Yrange);
        
        \draw[thick, orange!70, dashed] (0,0) circle (\invNorm);

        \draw[->] (-\Xrange-0.1, 0) -- (\Xrange+0.1, 0) node[right] {$\Re z$};
        \draw[->] (0, -\Yrange-0.1) -- (0, \Yrange+0.1) node[above] {$\Im z$};
        
        \node[below right] at (0,0) {$0$};

        \draw[dashed, thick, blue!60] (-\invNorm, -\Yrange) -- (-\invNorm, \Yrange);

        \filldraw (-\invNorm, 0) circle (1pt);
        \node[above left, xshift=2pt] at (-\invNorm, 0) {$-\|T^{-1/2}g\|^{-2}$};
        
        \node[black, font=\Large] at (2.0, 1.5) {$U$};
        
    \end{tikzpicture}
    \caption{The region $U$ for the single pair interaction model. 
      The shaded area represents the half-plane $\Re(z) > -\|T^{-1/2}g\|^{-2}$, 
      and the dashed circle indicates the convergence disk with radius $\|T^{-1/2}g\|^{-2}$.}
    \label{fig:region_U_spm}
\end{figure}

\subsection{\texorpdfstring{The Pauli--Fierz model with $x^2$-potentials in the dipole approximation}{The Pauli--Fierz model with x\texttwosuperior-potentials (dipole approx.)}}\label{sec:DPF}

We apply the general theory of holomorphy developed in Sections \ref{sect:GSE} and \ref{sect:GS} 
to the Pauli--Fierz model in the dipole approximation.

We consider a system in which an electron, bound to the origin by a spring, interacts with the 
quantized radiation field.
We prove that the ground state and its energy are holomorphic 
with respect to the charge in a complex strip-shaped domain that includes the real axis.
Consequently, the radius of convergence of the perturbative expansion around the origin 
is at least half the width of this strip.

Furthermore, when the electromagnetic interaction is cut off at the Compton wavelength
(used in the derivation of the Lamb shift), 
it is shown that the ground state and its energy remain holomorphic over a domain that 
includes the value of the elementary charge. 
We consider both cases: with and without mass renormalization, and compare the corresponding radii of convergence.
In either case, the perturbation expansion is justified for the physically relevant value of the charge.

In this section, we adopt the SI unit system. The symbols $\hbar$, $c$, $e$, $\ep_0$, and $m$
 will represent the reduced Planck constant, the speed of light, the elementary charge, the vacuum permittivity,
 and the physical mass of the electron, respectively.

\subsubsection{Definition of the Model}
We consider a system consisting of a nonrelativistic electron in a $d$-dimensional space
with $d\geq 2$, connected to springs with spring constants $(\kappa_1,\cdots,\kappa_d) \in \RR_{>0}^d$ 
at the origin, interacting with a quantized radiation field in the dipole approximation.

The Hilbert space for the system is $L^2(\RR_\bx^d)\tensor \Fb(L^2(\RR_\bk^d\times \ZZ_{d-1}))$,
where $\ZZ_{d-1}\coloneqq  \{1,2,\cdots,d-1\}$.
The Hamiltonian is defined by
\begin{align} \label{hamil:dpf}
  H(q) \coloneqq  \frac{1}{2m_0} \left(\bp\tensor \one - q\one\tensor \bA(\boldsymbol{0})\right)^2 
          + \one \tensor \dGb(\hbar c|\bk|) +  \frac{1}{2} \sum_{j=1}^d \kappa_j x_j^2,
\end{align}
where $m_0>0$ is the ``mass'' of the electron (which may refer to either 
the observed mass $m$ or the bare mass), $\bp=-i\hbar \nabla_\bx$ is 
the momentum of the electron, and $\bA(\boldsymbol{0})$ is the quantized
vector potential at the origin.
The constant $q$ is physically the electron charge $-e$, but it is treated as a parameter here.

Next, we define $\bA(\boldsymbol{0})$.
We fix polarization vectors $\be^{(\lambda)}:\RR^d\to \RR^d$, $\lambda\in\ZZ_{d-1}$, such that
$\be^{(\lambda)}(\bk)$ is measurable and
$\{\bk/|\bk|, \be^{(1)}(\bk), \cdots, \be^{(d-1)}(\bk)\}$
forms an orthonormal basis of $\RR^d$ for almost every $\bk\in\RR^d$.
For a function $\hat\rho(\bk)$, we define the coupling function 
$\boldsymbol{g}=(g_1,\cdots,g_d)$ as follows:
\begin{align}
  g_j(\bk,\lambda) = \sqrt{\frac{\hbar}{\ep_0 c}} e^{(\lambda)}_j(\bk) |\bk|^{-1/2}\hat\rho(\bk), 
  \qquad (\bk,\lambda) \in \RR^d\times \ZZ_{d-1}.   \label{eq:def_gj}
\end{align}
Then the quantized vector potential at the origin is given by
\begin{align}
 \bA(\boldsymbol{0}) \coloneqq   \PhiS(\boldsymbol{g}) 
                     \coloneqq  (\PhiS(g_1),\cdots , \PhiS(g_d)).
\end{align}
The function $\hat{\rho}$ is referred to as the ultraviolet cutoff.

In the following, we assume that $\hat\rho$ satisfies the following conditions:
\begin{itemize}
\item[(PF1)] $\hat{\rho}(\bk)$ is real-valued and spherically symmetric.
\item[(PF2)] $\hat{\rho}, ~ |\bk|^{-1} \hat{\rho} \in L^2(\RR^d)\setminus\{0\}$.
\end{itemize}

The most important example of an ultraviolet cutoff is the sharp cutoff at the Compton 
wave number $k_\mathrm{c} \coloneqq  mc/\hbar$, given by  
\begin{align}  
  \hat{\rho}(\bk) = (2\pi)^{-d/2} \chi(|\bk| < k_\mathrm{c}), \label{std uv}  
\end{align}
under which the Lamb shift was derived in \cite{PhysRev.72.339}.

In the Pauli--Fierz model, it is known that the electromagnetic interaction increases the mass of the electron.
This phenomenon is known as mass renormalization.
Therefore, the mass $m_0$ of the electron appearing in the Hamiltonian must be the value 
obtained by subtracting this increase, which is referred to as the bare mass $m_\mathrm{bare}$.
We will consider both the case where $m_0$ is the bare mass $m_\mathrm{bare}$ and the case where $m_0$
is the observed mass $m$ without taking mass renormalization into account.

The relationship between $m$ and $m_\mathrm{bare}$ is
\begin{align}
 m = m_\mathrm{bare} + \delta m.
\end{align}
In the dipole approximation, the mass renormalization term $\delta m$ is exactly given by
\begin{align}
\delta m \coloneqq  \norm{(\hbar c |\bk|)^{-1/2} eg_1}^2 
= \frac{e^2}{\ep_0 c^2} \frac{d-1}{d} \norm{|\bk|^{-1} \hat\rho}^2_{L^2(\RR^d)},  
\end{align}
as established in \cite{MR700181, MR1877237} (see also \cite[Theorem 7.6]{MR4213757}).
Therefore, the bare mass must be defined as
\begin{align}
  m_\mathrm{bare}
  & \coloneqq  m - \delta m \\
  & = m - \frac{e^2}{\ep_0 c^2}\frac{d-1}{d} \norm{|\bk|^{-1}\hat\rho}^2_{L^2(\RR^d)}, \label{def:mbare}
\end{align}
where the right-hand side is required to be positive.
When adopting the ultraviolet cutoff \eqref{std uv} for $ d=3 $, we have
\begin{align}
  \delta m = \frac{4\hbar}{3\pi c} k_\mathrm{c} \alpha = \frac{4\alpha}{3\pi} m \approx \frac{1}{322.9}m,
  \label{mbare}
\end{align}
which implies that
\begin{align}
  m_\mathrm{bare} = \left(1-\frac{4\alpha}{3\pi} \right) m, \label{eq:bare_mass_relation}
\end{align}
where $\alpha = (4\pi \ep_0)^{-1}e^2/\hbar c \approx 1/137$ is the fine-structure constant.
Since $\delta m$ arises from electromagnetic interactions, it depends on the charge $q$. 
However, we will treat it as a constant by using the value obtained with $q = -e$.

We note that this model is independent of the choice of polarization vectors. 
In other words, the Hamiltonians obtained using different polarization vectors are 
unitarily equivalent (see \cite[Appendix A]{MR3223476}, where it should be noted 
that the assumption $\hat\rho(\bk) = \hat\rho(-\bk)$ is implicitly made).

\subsubsection{Holomorphy of the Ground State and Ground State Energy}
We will transform the Hamiltonian \eqref{hamil:dpf} into the bosonic quadratic Hamiltonian \eqref{HAMIL} and analyze it. 
The procedure is described in detail in \cite[Section 7.3]{MR4213757}, but since there are slight differences in notation and unit systems compared to the current work, we will reiterate the necessary parts.

Let $\{a_j^*, a_j\}_{j=1}^d$ be the standard creation and annihilation operators on $L^2(\RR_\bx^d)$, that is,
\begin{align}
  x_j = \sqrt{\frac{\hbar}{2m_0\ome_j}}\ovl{(a_j + a_j^*)}, \qquad 
  p_j = \sqrt{\frac{m_0\hbar \ome_j}{2}}(\ovl{ia_j^* -i a_j}),
\end{align}
where $\ome_j \coloneqq  \sqrt{\kappa_j/m_0}$.
There exists a natural unitary operator from $L^2(\RR_\bx^d)$ onto $\Fb(\CC^d)$ such that 
\begin{align}
  & \sum_{j=1}^d \hbar \ome_j \left(a_j^* a_j + \frac{1}{2}\right) 
    \sim \dGb(\hbar\ome) + \sum_{j=1}^d \frac{\hbar \ome_j}{2}, \\
  & x_j \sim \sqrt{\frac{\hbar}{m_0\ome_j}} \PhiS(v_j), \qquad 
    p_j \sim \sqrt{m_0\hbar \ome_j}\PhiS( iv_j),
\end{align}
where $\ome = \diag(\ome_1,\cdots, \ome_d)$ is a diagonal matrix, and $\{v_j\}_{j=1}^d$ is the 
standard basis of $\CC^d$.
Under this unitary operator, $H(q)$ can be identified with 
\begin{align}
  \dGb(\hbar\ome)\tensor \one + \sum_{j=1}^d \frac{\hbar \ome_j}{2}
 + \one \tensor \dGb(\hbar c|\bk|)
 - \frac{q}{m_0}\sum_{j=1}^d \sqrt{m_0\hbar \ome_j} \PhiS(iv_j)\tensor \PhiS(g_j)\\
 + \frac{q^2}{2m_0}\sum_{j=1}^d \one\tensor \PhiS(g_j)^2.
\end{align}
By the unitary transformation $\Gamma(i\one) \tensor \one$, the term $\PhiS(iv)$ becomes $-\PhiS(v)$.
 Therefore, the Hamiltonian can be identified with
\begin{align}
  \dGb(\hbar\ome) \tensor \one + \sum_{j=1}^d \frac{\hbar \ome_j}{2} + \one \tensor \dGb(\hbar c|\bk|) 
  + \frac{q}{m_0} \sum_{j=1}^d \sqrt{m_0 \hbar \ome_j} \PhiS(v_j) \tensor \PhiS(g_j) \\ 
  + \frac{q^2}{2m_0} \sum_{j=1}^d \one \tensor \PhiS(g_j)^2.
\end{align}
Furthermore, by using the natural isomorphism
\begin{align}
 \Fb(\CC^d)\tensor \Fb(L^2(\RR^d\times\ZZ_{d-1})) \cong \Fb(\CC^d\oplus L^2(\RR^d\times \ZZ_{d-1})) =: \sF,
\end{align}
the Hamiltonian $H(q)$ can be expressed in the form of the bosonic quadratic Hamiltonian on $\sF$:
\begin{align}
  \til{H}(q) \coloneqq  \dGb(T) + \frac{1}{2}\sum_{j=1}^d \PhiS(0,q\til{g}_j)^2
  + \frac{1}{4} \sum_{j=1}^d (\hbar\ome_j)^{1/2} \big[ \PhiS(v_j,q\til{g}_j)^2 - \PhiS(v_j,-q\til{g}_j)^2\big]
  + E_0,
\end{align}
where 
\begin{align}
   T \coloneqq  \hbar \ome \oplus \hbar c |\bk|, \qquad 
   \til{g}_j \coloneqq  \frac{g_j}{\sqrt{m_0}}, \qquad 
   E_0 \coloneqq   \sum_{j=1}^d \frac{\hbar \ome_j}{2}.
\end{align}
$H(q)$ and $\til{H}(q)$ are equivalent under a unitary transformation that does not depend on $q$.
Therefore, to study the holomorphy of the ground state and ground state energy of $H(q)$ with respect to $q$, it suffices to study those of $\tilde{H}(q)$.

Corresponding to \eqref{defW}, we define
\begin{align}
  W(z) \coloneqq  \sum_{j=1}^d
  \begin{pmatrix}
    0 & z \ketbra{\hbar \ome_j v_j}{(\hbar c|\bk|)^{1/2}\til{g}_j} \\
   z \ketbra{(\hbar c |\bk|)^{1/2}\til{g}_j}{\hbar\ome_j v_j} & z^2 \ketbra{(\hbar c |\bk|)^{1/2}\til{g}_j}{(\hbar c |\bk|)^{1/2}\til{g}_j}
  \end{pmatrix}.
\end{align}
The operator $S(q)$ is defined by  
\begin{align}  
  S(q) \coloneqq  \big( T^2 + W(q) \big)^{1/2}.  
\end{align}

\begin{Theorem}[{\cite[Theorem 7.4]{MR4213757}}]\label{diagonalization of HODPF}
For any $q\in\RR$, $H(q)$ is self-adjoint, and essentially self-adjoint on any core of $\dGb(T)$.
Furthermore, the unitary operator $\sU=\sU(q)$ on $\Fb(\sH)$ satisfying equation \eqref{defU*} yields
\begin{equation}
   \sU H(q) \sU^* = \dGb\left( S(q)\right) + E(q), \qquad 
   E(q) = \frac{1}{2}\tr\left(\,\ovl{S(q)-T}\,\right) + E_0.
\end{equation}
\end{Theorem}

\begin{proof}
Since $\hat\rho(\bk)$ is real-valued, the conjugation operator $J$ on $L^2(\RR^d \times \ZZ_{d-1})$ 
defined by  
\begin{align}  
    (Jf)(\bk,\lambda) \coloneqq  \ovl{f(\bk,\lambda)}, \qquad f \in L^2(\RR^d \times \ZZ_{d-1}),
                         ~(\bk,\lambda) \in \RR^d\times \ZZ_{d-1}
\end{align}
allows us to apply the proof in \cite[Theorem 7.4]{MR4213757} without modification.
\end{proof}

The pair $(T,W)$ satisfies conditions (C1)--(C5) with $U=\CC$.
Since condition (C6) may not hold for $U=\CC$, we set
\begin{align}
  U \coloneqq  \{ z \in\CC \mid |\Im z| < Q \},   \label{def of U_pf}
\end{align}
where $Q$ is a quantity with the dimension of charge, given by
\begin{align}
Q &\coloneqq  \norm{(\hbar c|\bk|)^{-1/2} \til{g}_1}^{-1} \\
  & = \left( \frac{d}{d-1}\ep_0 c^2 m_0 \right)^{1/2} \norm{|\bk|^{-1} \hat\rho}_{L^2(\RR^d)}^{-1}.  \label{def of Q}
\end{align}
We note that 
\begin{align}
  \frac{e^2}{Q^2} = \frac{\delta m}{m_0}.
\end{align}

\begin{Proposition}
Let $z$ be a complex number. 
Then $\Re( 1 +\ovl{T^{-1}W(z)T^{-1}}) >0$ holds if and only if $z\in U$.
That is, the pair $(T,W)$ satisfies conditions (C1)--(C6) with $U$ defined by \eqref{def of U_pf}.
In particular, $S(z)$ can be defined for all $z \in U$.
\end{Proposition}

\begin{proof}
Set $\tau_j \coloneqq  (\hbar c|\bk|)^{-1/2}\til{g}_j$. Then
\begin{align}
  \ovl{T^{-1}W(z)T^{-1}} = 
  \sum_{j=1}^d 
  \begin{pmatrix}
    0 & z \ketbra{v_j}{\tau_j} \\
   z \ketbra{\tau_j}{v_j} & z^2 \ketbra{\tau_j}{\tau_j}
  \end{pmatrix}.
\end{align}
Consider the operator
\begin{align}
     M(z) \coloneqq  \ovl{T^{-1}W(z)T^{-1}} + (\ovl{T^{-1}W(z)T^{-1}})^*,
    \qquad z\in\CC.
\end{align}
It is enough to show that $M(z)>-2$ if and only if $z\in U$.

From the properties of the polarization vectors and the spherical symmetry of $\hat{\rho}$, it follows that  
\begin{align}  
    \inner{\tau_i}{\tau_j} = \norm{\tau_1}^2 \delta_{ij}.  
\end{align}  
Therefore,  
\begin{align}  
    V_j \coloneqq  \mathrm{Span} \{ (v_j,0), (0,\tau_j) \}, \qquad j=1,\cdots,d  
\end{align}  
are mutually orthogonal subspaces.
Moreover, $M(z)$ is reduced by each $V_j$, and the reduced part is given by  
\begin{align}  
    M_j(z) \coloneqq  \begin{pmatrix}  
        0 & (z+\bar{z}) \ketbra{v_j}{\tau_j} \\  
        (z+\bar{z}) \ketbra{\tau_j}{v_j} & (z^2+\bar{z}^2) \ketbra{\tau_j}{\tau_j}  
    \end{pmatrix}.  
\end{align}
The matrix representation of $M_j(z)$ with respect to the basis $(v_j, 0)$ and $(0, \tau_j / \norm{\tau_j})$ is  
\begin{align}  
    \til{M}_j(z) \coloneqq   
    \begin{pmatrix}  
        0 & (z+\bar{z}) \norm{\tau_1} \\  
        (z+\bar{z}) \norm{\tau_1} & (z^2+\bar{z}^2) \norm{\tau_1}^2  
    \end{pmatrix}.  
\end{align}  
The eigenvalues of this matrix are  
\begin{align}  
    \lambda_\pm(z) = \frac{(z^2+\bar{z}^2)\norm{\tau_1}^2 \pm \sqrt{(z^2+\bar{z}^2)^2\norm{\tau_1}^4 
  + 4(z+\bar{z})^2\norm{\tau_1}^2}}{2}.  
\end{align}  
Therefore,  
\begin{align}  
    \inf \sigma(M(z)) = \lambda_-(z).  
\end{align}  
A simple calculation shows that $\lambda_-(z) > -2$ if and only if
\begin{align}  
    |\Im z| < \norm{\tau_1}^{-1} = Q.  
\end{align}
\end{proof}

Based on the above proposition and the discussions in Section~\ref{sect:GSE}, 
the ground state energy $E(q)$ can be extended to a holomorphic function on $U$. 
Moreover, by the discussions in Section \ref{sect:GS}, $K_1(q)$ can be analytically continued to 
a holomorphic function on $U$, and therefore, the ground state $\Phigt(q)$ can also 
be extended to a holomorphic function on $U$.
Note that $\Phigt(z)$ is defined in \eqref{extGS} and 
\begin{align}
  E(z) \coloneqq  \frac{1}{2} \tr \ovl{(S(z)-T)|_{\dom(T^2)}} + E_0.
\end{align}

Consequently, we obtain the following theorem:
\begin{Theorem}\label{thm:dpf}  
The ground state energy $E(q)$ and the ground state $\til{\Phi}_\mathrm{g}(q)$ of $H(q)$ can be
extended to holomorphic functions on the complex domain $U$.  
In particular, the Maclaurin series of $E(q)$ and $\til{\Phi}_\mathrm{g}(q)$ converge absolutely for $|q| < Q$.  
\end{Theorem}

\begin{Example}\label{example:Q2}
Consider a physical situation where $d=3$ and the ultraviolet cutoff is given by \eqref{std uv}.
In the case where $m_0 = m_\mathrm{bare} = (1 - \frac{4\alpha}{3\pi})m$, we have
\begin{align}
  Q = \left(\frac{3\pi - 4\alpha}{4\alpha}\right)^{1/2} e \approx 17.941 e,   \label{eq:Q_d3_bare}
\end{align}
(see below for details of the calculation).
Therefore, $-e \in U$, and the Maclaurin series of the ground state and ground state energy of
 $H(-e)$ with respect to the elementary charge $e$ converges absolutely.

In the case where $m_0=m$, we have
\begin{align}
  Q = \left(\frac{3\pi}{4\alpha}\right)^{1/2}e \approx  17.969 e,
\end{align}
so the convergence disk becomes slightly larger.

Figure \ref{fig:region_U_x2} illustrates the relationship between the region $U$, the disk of radius $Q$,
 and the electron charge $-e$. In either case, the electron charge is contained within this disk.

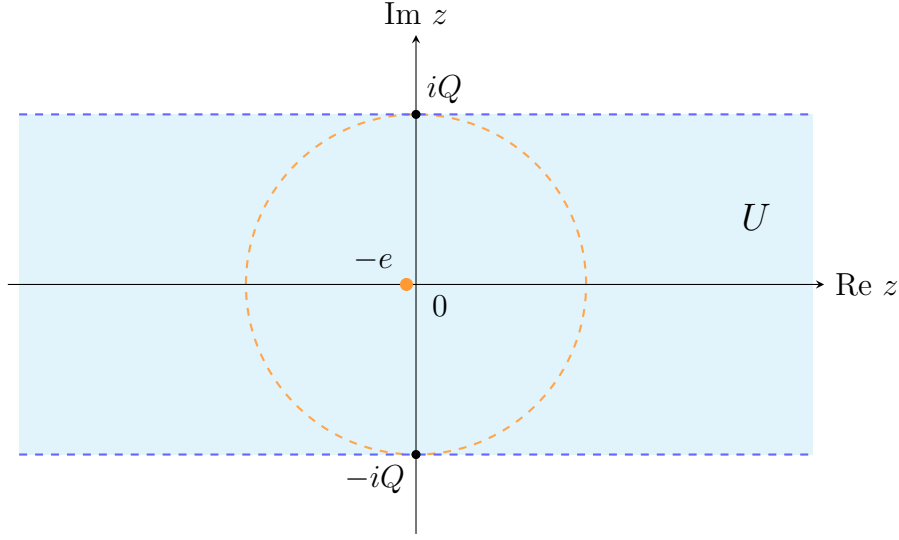
\begin{figure}[htbp]
    \centering
    \begin{tikzpicture}[>=stealth, scale=1.5]
        \def\Q{1.5}       
        \def\Xrange{3.5}  
        \def\minusE{-0.084} 

        \fill[cyan!10] (-\Xrange, -\Q) rectangle (\Xrange, \Q);
        
        \draw[thick, orange!70, dashed] (0,0) circle (\Q);

        \draw[->] (-\Xrange-0.1, 0) -- (\Xrange+0.1, 0) node[right] {$\text{Re } z$};
        \draw[->] (0, -\Q-0.7) -- (0, \Q+0.7) node[above] {$\text{Im } z$};
        
        \node[below right, xshift=2pt] at (0,0) {$0$};

        \draw[dashed, thick, blue!60] (-\Xrange, \Q) -- (\Xrange, \Q);
        \draw[dashed, thick, blue!60] (-\Xrange, -\Q) -- (\Xrange, -\Q);

        \filldraw (0, \Q) circle (1pt);
        \node[above right] at (0, \Q) {$iQ$};
        
        \filldraw (0, -\Q) circle (1pt);
        \node[below left] at (0, -\Q) {$-iQ$};

        \node[black, font=\large] at (\Xrange-0.5, 0.6) {$U$};
        
        \filldraw[orange!80] (\minusE, 0) circle (1.5pt);
        \node[above left, black, xshift=-1pt, yshift=1pt] at (\minusE, 0) {$-e$};

    \end{tikzpicture}
    \caption{The region $U$ and the point $-e$ in the physical situation described in Example \ref{example:Q2}, where $Q \approx 17.9e$. 
      The dashed circle represents the disk $|z| < Q$.}
    \label{fig:region_U_x2}
\end{figure}

We describe the details of the calculation for $Q=\norm{\tau_1}^{-1}$.
Recall that
\begin{align}
  \tau_1(\bk,\lambda)
  = (\hbar c |\bk|)^{-1/2} \frac{g_1(\bk,\lambda)}{\sqrt{m_0}} 
  = \left( \frac{1}{m_0 \ep_0 c^2}\right)^{1/2} e_1^{(\lambda)}(\bk) |\bk|^{-1}\hat\rho(\bk).
\end{align}
Since $\bk/|\bk|, \be^{(1)}(\bk), \cdots, \be^{(d-1)}(\bk)$ form an orthonormal basis
and $\hat\rho$ is spherically symmetric, we obtain
\begin{align}
  \norm{\tau_1}^2
 = \frac{1}{m_0 \ep_0c^2} \sum_{\lambda=1}^{d-1} \int_{\RR^d} e_1^{(\lambda)}(\bk)^2 |\bk|^{-1}|\hat\rho(\bk)|^2 \, d\bk
 = \frac{1}{m_0 \ep_0c^2}\left(1-\frac{1}{d}\right) \norm{|\bk|^{-1}\hat\rho}^2_{L^2(\RR^d)}.
\end{align}
For the case $d=3$ with the ultraviolet cutoff given by \eqref{std uv}, we have
\begin{align}
  \norm{|\bk|^{-1}\hat\rho}^2_{L^2(\RR^3)} 
 = \int_{\RR^3} (2\pi)^{-3} |\bk|^{-2}\chi(|\bk|<k_\mathrm{c}) \, d\bk
 = \frac{1}{2\pi^2} k_\mathrm{c}.
\end{align}
Thus
\begin{align}
  \norm{\tau_1}^2 
  = \frac{1}{3\pi^2} \frac{1}{m_0\ep_0c^2} k_\mathrm{c}
  =  \frac{4}{3\pi} \frac{m}{m_0}\frac{e^2}{4\pi\ep_0c\hbar } \frac{1}{e^2}
  =  \frac{4}{3\pi} \frac{m}{m_0} \alpha \frac{1}{e^2}.
\end{align}
In the case where $m_0=m_\mathrm{bare}=m-\delta m$, it follows that 
$\norm{\tau_1}^2 = 4\alpha (3\pi-4\alpha)^{-1} e^{-2}$
which leads to 
\begin{align}
  Q = \norm{\tau_1}^{-1} = \left( \frac{3\pi - 4\alpha}{4\alpha}\right)^{1/2} e.
\end{align}
On the other hand, in the case where $m_0=m$, since $\norm{\tau_1}^2 = (4\alpha/3\pi)e^{-2}$,
we obtain
\begin{align}
  Q = \left( \frac{3\pi}{4\alpha}\right)^{1/2} e. 
\end{align}
\end{Example}

The above theorem shows that the radii of convergence of the Maclaurin series 
for the ground state and its energy are greater than or equal to $ Q $.
We conclude this subsection by noting that $Q$ is likely not the actual value of the radii of convergence:
\begin{Conjecture}\label{conj:x2}
The radii of convergence of the Maclaurin series for the ground state and its energy are finite values strictly greater than $Q$.
\end{Conjecture}

\subsection{The translation-invariant Pauli--Fierz model}\label{sec:TDPF}
We apply the general theory of holomorphy developed in Sections \ref{sect:GSE} and \ref{sect:GS} 
to the translation-invariant Pauli--Fierz model in the dipole approximation. 

\subsubsection{Hamiltonian and fiber decomposition}
We consider a free electron interacting with a quantized radiation field in the dipole approximation.
The Hamiltonian of this system is obtained from \eqref{hamil:dpf} by setting $(\kappa_1,\cdots,\kappa_d) = (0,\cdots,0)$:
\begin{align} \label{hamil:tdpf}
  H(q)
  \coloneqq  \frac{1}{2m_0}
     \left(\bp \tensor \one - q\,\one \tensor \bA(\boldsymbol{0})\right)^{2}
     + \one \tensor \dGb(\hbar c |\bk|).
\end{align}
Let $\sF_d:L^2(\RR_\bx^d)\to L^2(\RR_\bP^d)$ be the Fourier transform defined by
\begin{align}
  (\sF_d f)(\bP) \coloneqq  \frac{1}{(2\pi\hbar)^{d/2}} \int_{\RR^d} f(\bx)e^{-i\bP\cdot\bx/\hbar} \, d\bx, \qquad 
  f\in L^2(\RR_\bx^d), ~ \bP \in \RR^d.
\end{align}
We use the natural isomorphism
\begin{align}
  L^2(\RR_\bP^d)\tensor \Fb(L^2(\RR_\bk^d\times \ZZ_{d-1})) \cong
  \int_{\RR^d}^\oplus \Fb(L^2(\RR_\bk^d\times \ZZ_{d-1})) \, d\bP.
\end{align}
Then we obtain the operator identity
\begin{align}
  (\sF_d\tensor \one) H(q) (\sF_d\tensor \one)^* = \int_{\RR^d}^\oplus H(q,\bP) \, d\bP,
\end{align}
where 
\begin{align}
  H(q,\bP) 
  & \coloneqq  \frac{1}{2m_0} \left(\bP - q \bA(\boldsymbol{0})\right)^2 + \dGb(\hbar c|\bk|) \\
  & = \dGb(\hbar c|\bk|) + \frac{q^2}{2m_0} \sum_{j=1}^d \PhiS(g_j)^2
      - \frac{q}{m_0} \sum_{j=1}^d P_j \PhiS(g_j)+\frac{\bP^2}{2m_0}.
\end{align}
The first two terms
\begin{align}
  H(q,\boldsymbol{0}) = \dGb(\hbar c|\bk|) + \frac{q^2}{2m_0}\sum_{j=1}^d \PhiS(g_j)^2
\end{align}
have the form of a bosonic quadratic Hamiltonian and can be diagonalized by a Bogoliubov transformation.

Corresponding to \eqref{HAMIL} and \eqref{defW}, we define
\begin{align}
     T & \coloneqq  \hbar c|\bk|, \\
  W(z) & \coloneqq  \frac{z^2}{m_0} \sum_{j=1}^d \ketbra{T^{1/2}g_j}{T^{1/2}g_j}, \qquad z\in\CC.
\end{align}
We set 
\begin{align}
  S(q) \coloneqq  \sqrt{T^2 + W(q)}, \qquad q\in\RR.
\end{align}
The following theorem is known.
\begin{Theorem}[{\cite[Section 7.4]{MR4213757}}]\label{diagonalization of TDPF}
  For any $q\in\RR$, $H(q,\boldsymbol{0})$ is self-adjoint and essentially 
self-adjoint on any core of $\dGb(\hbar c|\bk|)$.
Furthermore, the unitary operator $\sU=\sU(q)$ on $\Fb(\sH)$ that satisfies the equation \eqref{defU*} yields
\begin{align}
  \sU H(q,\boldsymbol{0}) \sU^* 
   = \dGb\left( S(q)\right) + E(q,\boldsymbol{0}), \qquad 
   E(q,\boldsymbol{0}) = \frac{1}{2} \tr \big(\ovl{S(q)-T}\big).
 \label{eq:diag_H(q,0)}
\end{align}
\end{Theorem}

By the theorem above, the (unnormalized) ground state $\til{\Phi}_{\mathrm{g}}(q,\boldsymbol{0})$
of $H(q,\boldsymbol{0})$ exists and is unique.  

\subsubsection{\texorpdfstring{The case $\bP=\boldsymbol{0}$}{The case P=0}}
We discuss the holomorphy in $q$ of $\til{\Phi}_{\mathrm{g}}(q,\boldsymbol{0})$ and $E(q,\boldsymbol{0})$.
The pair $(T,W)$ satisfies conditions (C1)--(C5) with $U = \CC$.
We now introduce the domain
\begin{align} \label{def of UTDPF}
  U \coloneqq  \{\, z \in \CC \mid \Re(z^{2}) > -Q^{2} \,\},
\end{align}
where $Q$ is defined in \eqref{def of Q}.
This is a domain bounded by the hyperbola $\Re(z^{2}) = -Q^{2}$ (see Figure \ref{fig:U_TDPF}).

\begin{Proposition}\label{prop:tdpf0cs}
Let $z$ be a complex number. 
Then $\Re( 1 +\ovl{T^{-1}W(z)T^{-1}}) >0$ holds if and only if $z\in U$.
That is, the pair $(T,W)$ satisfies conditions (C1)--(C6) with $U$ defined by \eqref{def of UTDPF}.
In particular, $S(z)$ can be defined for all $z \in U$.
\end{Proposition}
\begin{proof}
First, we note that
\begin{align}
  \ovl{T^{-1} W(z) T^{-1}}
  = \frac{z^{2}}{m_0} \sum_{j=1}^{d} \ketbra{T^{-1/2} g_j}{T^{-1/2} g_j},
  \qquad z \in \CC.
\end{align}
Since $\hat{\rho}$ is spherically symmetric, the properties of the polarization vectors
imply that the vectors $T^{-1/2} g_j$ form an orthogonal system, and
\begin{align}
  \inner{T^{-1/2} g_j}{T^{-1/2} g_\ell}
  = \delta_{j,\ell}\, \norm{T^{-1/2} g_j}^{2}.
\end{align}
Therefore, the condition
\begin{equation}
  \Re \big( 1 + \ovl{T^{-1} W(z) T^{-1}} \big) > 0
\end{equation}
is equivalent to
\begin{align}
  \frac{\Re(z^{2})}{m_0} \ketbra{T^{-1/2} g_j}{T^{-1/2} g_j} > -1,
  \qquad j = 1,\cdots,d.
\end{align}
This is of course equivalent to
\begin{align}
  \Re(z^{2})
    > - m_0 \norm{T^{-1/2} g_j}^{-2}
    = -Q^{2},
  \qquad j = 1,\cdots,d.
\end{align}
Thus, condition \textnormal{(C6)} holds.
\end{proof}

Based on the above proposition and the discussions in Section \ref{sect:GSE}, 
the ground state energy $E(q,\boldsymbol{0})$ can be extended to a holomorphic function on $U$. 
Moreover, by the discussions in Section \ref{sect:GS}, $K_1(q)$ can be analytically continued to a
holomorphic function on $U$, and therefore, the ground state $\Phigt(q,\boldsymbol{0})$ 
can also be extended to a holomorphic function on $U$.
Note that $\Phigt(z,\boldsymbol{0})$ is defined in \eqref{extGS} and 
\begin{align}
  E(z,\boldsymbol{0}) \coloneqq  \frac{1}{2} \tr \ovl{(S(z)-T)|_{\dom(T^2)}}.
\end{align}
Thus, we obtain the following theorem:
\begin{Theorem}\label{thm:tdpf}  
The ground state energy $E(q,\boldsymbol{0})$ and the ground state $\til{\Phi}_\mathrm{g}(q,\boldsymbol{0})$ 
of $H(q,\boldsymbol{0})$ can be extended as holomorphic functions over the complex domain $U$.
In particular, the Maclaurin series of $E(q,\boldsymbol{0})$ and $\til{\Phi}_\mathrm{g}(q,\boldsymbol{0})$ 
converge absolutely for $|q| < Q$.  
\end{Theorem}

\begin{Remark}
Under the same physical conditions as in Example \ref{example:Q2} (i.e., $d=3$, with the ultraviolet cutoff taken as 
in \eqref{std uv} and $m_0=m-\delta m$), the region $U$ is as illustrated in Figure \ref{fig:U_TDPF}.
As will be shown below, $Q$ gives the actual radius of convergence for the ground state and the ground state energy.
Note that this is not expected to hold in the case of the $x^2$-potential (see Conjecture \ref{conj:x2}).
\end{Remark}

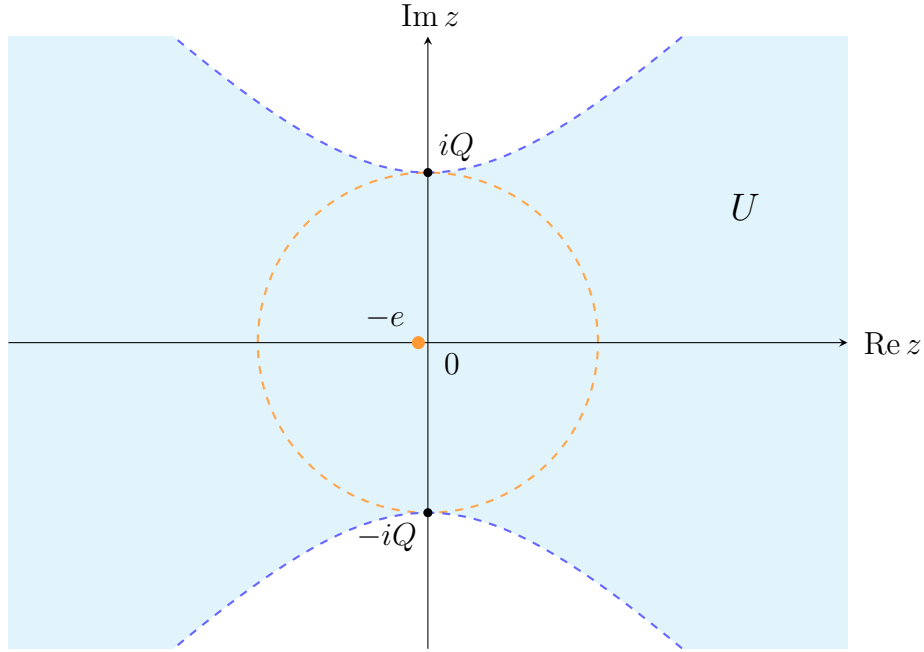
\begin{figure}[htbp]
    \centering
    \begin{tikzpicture}[>=stealth, scale=1.5]
        \def\Q{1.5}       
        \def\Xrange{3.5}  
        \def\Yrange{2.5} 
        \def\minusE{-0.084} 
        \def\offset{0.2}

        \begin{scope}
            \clip (-\Xrange-\offset, -\Yrange-\offset) rectangle (\Xrange+\offset, \Yrange+\offset);

            \fill[cyan!10] (-\Xrange-1, -\Yrange-1) rectangle (\Xrange+1, \Yrange+1);
            
            \fill[white] plot[domain=-\Xrange-1:\Xrange+1, samples=100] (\x, {sqrt(\x*\x + \Q*\Q)}) -- (\Xrange+1, \Yrange+5) -- (-\Xrange-1, \Yrange+5) -- cycle;
            \fill[white] plot[domain=-\Xrange-1:\Xrange+1, samples=100] (\x, {-sqrt(\x*\x + \Q*\Q)}) -- (\Xrange+1, -\Yrange-5) -- (-\Xrange-1, -\Yrange-5) -- cycle;

            \draw[thick, orange!70, dashed] (0,0) circle (\Q);

            \draw[dashed, thick, blue!60, domain=-\Xrange-1:\Xrange+1, samples=100] 
                plot (\x, {sqrt(\x*\x + \Q*\Q)});
            \draw[dashed, thick, blue!60, domain=-\Xrange-1:\Xrange+1, samples=100] 
                plot (\x, {-sqrt(\x*\x + \Q*\Q)});
        \end{scope}

        \draw[->] (-\Xrange-\offset, 0) -- (\Xrange+\offset, 0) node[right] {$\Re z$};
        \draw[->] (0, -\Yrange-\offset) -- (0, \Yrange+\offset) node[above] {$\Im z$};
        \node[below right, xshift=2pt] at (0,0) {$0$};

        \filldraw[orange!80] (\minusE, 0) circle (1.5pt);
        \filldraw (0, \Q) circle (1pt);
        \node[above right] at (0, \Q) {$iQ$};
        \filldraw (0, -\Q) circle (1pt);
        \node[below left] at (0, -\Q) {$-iQ$};
        
        \node[black, font=\large] at (\Xrange-0.7, 1.2) {$U$};
        \node[above left, black, xshift=-1pt, yshift=1pt] at (\minusE, 0) {$-e$};
    \end{tikzpicture}
\caption{The region $U$ defined by $\Re(z^2) > -Q^2$ for the translation-invariant Pauli--Fierz model. 
Here $Q \approx 17.9e$ under physical conditions similar to those in Example \ref{example:Q2}.}
\label{fig:U_TDPF}
\end{figure}

From the general argument developed in Section~\ref{sect:GSE}, we can obtain only the above result.
However, we note that the ground state energy actually admits an analytic continuation to a larger domain,
 as stated in the following theorem.
\begin{Theorem}\label{thm:20261229}
The ground state energy $E(q,\boldsymbol{0})$ can be holomorphically extended to 
\begin{align}
  \CC \setminus \{\, bi \mid b \in \RR,\ |b| \geq Q \}.    \label{def:hiroiU}
\end{align}
\end{Theorem}

To prove this theorem, we use an integral representation of the ground state energy.
This representation was first obtained in \cite[Lemma~4.10]{MR1877237}.
While additional assumptions on $\hat\rho$ were originally imposed there (see also \cite[Theorem 3.32]{EB}),
one can readily check that the same representation holds under our present general setting.

\begin{Lemma}\label{lem:E(q,0)_int}
The ground state energy $E(q,\boldsymbol{0})$ is given by the following integral representation:
\begin{align}
  E(q, \boldsymbol{0}) = \frac{d}{\pi} \int_0^\infty
\frac{
\frac{d-1}{d}\frac{q^2\hbar^2}{m_0\ep_0}
 \bignorm{\frac{\hat\rho}{(\hbar c|\bk|)^2+t^2}}_{L^2(\RR^d)}^{2}
}{
1+\frac{d-1}{d}\frac{q^2\hbar^2}{m_0\ep_0}
 \bignorm{\frac{\hat\rho}{\sqrt{(\hbar c|\bk|)^2+t^2}} }_{L^2(\RR^d)}^{2}} \, t^2 dt .
\label{eq:E(q,0)_int}
\end{align}
\end{Lemma}

\begin{proof}
From \eqref{eq:diag_H(q,0)} and \eqref{eq:integ_repr_tr}, the ground state energy can be written as
\begin{align}
  E(q,\boldsymbol{0})
 = \frac{1}{\pi}\int_0^\infty \tr 
   \left[ R_t^{1/2} \big(1+\til{W}_t(q) \big)^{-1} R_t^{1/2}W(q)R_t \right]
   t^2\, d t.
\end{align}
Here, 
\begin{align}
  \til{W}_t(q) \coloneqq  R_t^{1/2}W(q)R_t^{1/2} 
             = \frac{q^2}{m_0} \sum_{j=1}^d \ketbra{R_t^{1/2}T^{1/2}g_j}{R_t^{1/2}T^{1/2}g_j}.
\end{align}
For notational simplicity, we set
\begin{align}
  \tau_j(t) \coloneqq  \frac{q}{ \sqrt{m_0} } R_t^{1/2}T^{1/2}g_j.
\end{align}
By the properties of the polarization vectors and the rotational symmetry of $\hat\rho$,
the family $\{ \tau_j(t) \}_{j=1}^d$ forms a mutually orthogonal system.
Consequently,
\begin{align}
  \big(1+\til{W}_t(q) \big)^{-1}
  = 1-\sum_{j=1}^d \frac{1}{ 1 + \norm{\tau_j(t)}^2}  \ketbra{ \tau_j(t) }{ \tau_j(t) }.
\end{align}
From this, we have
\begin{align}
&  \tr \left[ R_t^{1/2} \big(1+\til{W}_t(q) \big)^{-1} R_t^{1/2}W(q)R_t \right] \\
&  =   \tr \bigg[ R_t^{1/2} \bigg( 1-\sum_{j=1}^d \frac{1}{ 1 + \norm{\tau_j(t)}^2}
     \ketbra{ \tau_j(t) }{ \tau_j(t) } \bigg) 
     \sum_{\ell=1}^d \ketbra{ \tau_\ell(t)}{ \tau_\ell(t)}  R_t^{1/2} \bigg] \\
& = \sum_{\ell=1}^d \inner{R_t^{1/2}\tau_\ell(t)}{R_t^{1/2}\tau_\ell(t)}
    -\sum_{j=1}^d \frac{1}{1+\norm{\tau_j(t)}^2}\sum_{\ell=1}^d \inner{\tau_j(t)}{\tau_\ell(t)}
     \inner{R_t^{1/2}\tau_\ell(t)}{R_t^{1/2}\tau_j(t)}\\
& = \sum_{j=1}^d \norm{R_t^{1/2}\tau_j(t)}^2
    -\sum_{j=1}^d \frac{\norm{\tau_j(t)}^2}{1+\norm{\tau_j(t)}^2} 
     \norm{R_t^{1/2}\tau_j(t)}^2\\
& = \sum_{j=1}^d\frac{\norm{R_t^{1/2}\tau_j(t)}^2}{1+\norm{\tau_j(t)}^2}.
\end{align}
Since
\begin{align}
  \norm{\tau_j(t)}^2 
  & = \frac{q^2}{m_0} \norm{((\hbar c|\bk|)^2+t^2)^{-1/2} (\hbar c|\bk|)^{1/2} g_j}^2\\
  & = \frac{q^2}{m_0} \frac{\hbar}{\ep_0 c} \sum_{\lambda=1}^{d-1} \int_{\RR^d} 
      \frac{\hbar c|\bk| }{(\hbar c|\bk|)^2+t^2}  e_j^{(\lambda)}(\bk)^2 |\bk|^{-1} |\hat\rho(\bk)|^2 d\bk\\
  & =  \frac{d-1}{d}  \frac{q^2\hbar^2}{m_0\ep_0} \int_{\RR^d} 
      \frac{ |\hat\rho(\bk)|^2 }{(\hbar c|\bk|)^2+t^2}  d\bk, \\
 \norm{R_t^{1/2}\tau_j(t)}^2
  & =  \frac{d-1}{d}  \frac{q^2\hbar^2}{m_0\ep_0} \int_{\RR^d} 
      \frac{ |\hat\rho(\bk)|^2 }{((\hbar c|\bk|)^2+t^2)^2}  d\bk,
\end{align}
we obtain \eqref{eq:E(q,0)_int}.
\end{proof}

\begin{proof}[Proof of Theorem \ref{thm:20261229}]
We use Lemma \ref{lem:E(q,0)_int}.
Let $z_0=a+bi\,(a,b\in\RR)$ be a point in the domain \eqref{def:hiroiU}.
We show that the integral in \eqref{eq:E(q,0)_int} (with $q$ replaced by $z$) converges
 locally uniformly in a neighborhood of $z_0$.
To this end, we set
\begin{align}
  c_t\coloneqq  \frac{d-1}{d}\frac{\hbar^2}{m_0\ep_0}
  \bigg\| \frac{\hat\rho}{\sqrt{(\hbar c|\bk|)^2+t^2}} \bigg\|_{L^2(\RR^d)}^2. \label{def:c_t}
\end{align}
Then $0 \leq c_t \leq c_0 = Q^{-2}$ holds for all $t \geq 0$.
Define
\begin{align}
  \delta\coloneqq  \mathrm{dist}(z_0,\{bi \mid b\in\RR, |b|\geq Q\})>0.
\end{align}
Since $1 + z^2 c_t = 0$ implies $z \in \{bi \mid b \in \RR, |b| = c_t^{-1/2} \geq Q\}$, all zero points of the denominator lie on the
 set $\{\, bi \mid |b| \geq Q \}$.
Thus, by the compactness of the closed disk $\{z \in \CC \mid |z-z_0| \leq \delta/2\}$, there exists $\vep>0$ such that
\begin{align}
  |1+ z^2 c_t| \geq \vep \qquad (t>0)
\end{align}
for all $z$ satisfying $|z-z_0|\leq \delta/2$.

Consequently, the integrand in \eqref{eq:E(q,0)_int} is dominated uniformly in $z \in \{z\in\CC\mid |z-z_0| \leq \delta/2\,\}$ by
\begin{align}
  \frac{|z|^2}{\vep} \frac{d-1}{d}\frac{\hbar^2}{m_0\ep_0}
  \left\| \frac{\hat\rho}{(\hbar c|\bk|)^2+t^2} \right\|_{L^2(\RR^d)}^{2} t^2.
\end{align}
Integrating this bound with respect to $t$ yields
\begin{align}
  \int_0^\infty \left\| \frac{\hat\rho}{(\hbar c|\bk|)^2+t^2} \right\|_{L^2(\RR^d)}^{2} t^2 dt
  = \frac{\pi}{4\hbar c} \left\| |\bk|^{-1/2} \hat\rho \right\|_{L^2(\RR^d)}^2 < \infty.
\end{align}
By Lebesgue's dominated convergence theorem, $E(z,\boldsymbol{0})$ defines a holomorphic function on the domain \eqref{def:hiroiU}.
\end{proof}

\begin{Theorem}
The radii of convergence of the Maclaurin series of the ground state $\Phigt(q,\boldsymbol{0})$ 
and its energy $E(q,\boldsymbol{0})$ are $Q$.
\end{Theorem}

\begin{proof}
We first show that the radius of convergence of the Maclaurin series for $E(q,\boldsymbol{0})$ is $Q$.
By using Lemma \ref{lem:E(q,0)_int}, the Maclaurin series of $E(q,\boldsymbol{0})$ can be written,
using $c_t$ defined in \eqref{def:c_t}, as
\begin{align}
  \sum_{\ell=0}^\infty (-1)^\ell q^{2\ell+2} 
   \bigg\{ \frac{d}{\pi} \int_0^\infty
   \frac{d-1}{d} \frac{\hbar^2}{m_0\ep_0}
  \Bignorm{\frac{\hat\rho}{(\hbar c|\bk|)^2+t^2}}_{L^2(\RR^d)}^{2}
  c_t^\ell \, t^2dt \bigg\}.  \label{eq:EPF_div0}
\end{align}
For any $\vep>0$, set $q=i(1+\vep)Q$.
Then $(-1)$ times the above series is equal to
\begin{equation}
  \sum_{\ell =0}^\infty (Q(1+\vep))^{2\ell +2} 
  \frac{d}{\pi} \int_0^\infty
  \frac{d-1}{d}\frac{\hbar^2}{m_0\ep_0}
  \Bignorm{\frac{\hat\rho}{(\hbar c|\bk|)^2+t^2}}_{L^2(\RR^d)}^{2} c_t^\ell 
  \, t^2 dt.   \label{eq:EPF_div1}
\end{equation}
It is enough to show that this diverges to $+\infty$.

Since $\lim_{t\to +0} c_t = c_0 = Q^{-2}$, there exists $t_0>0$ such that
\begin{equation}
  (1+\vep)^2 c_{t_0}> Q^{-2}.
\end{equation}
Since $t \mapsto c_t$ is monotonically decreasing, we have $c_t \geq c_{t_0}$ for all $t \in [0, t_0]$.
Then
\begin{align}
  \eqref{eq:EPF_div1} 
& \geq \sum_{\ell =0}^\infty (Q(1+\vep))^{2\ell +2} \frac{d}{\pi} \int_0^{t_0}
  \frac{d-1}{d}\frac{\hbar^2}{m_0\ep_0}
  \biggnorm{\frac{\hat\rho}{(\hbar c|\bk|)^2+t^2}}_{L^2(\RR^d)}^{2}
  c_t^\ell  \, t^2 dt  \\
& \geq \sum_{\ell =0}^\infty (Q(1+\vep))^{2\ell +2} \frac{d}{\pi} \int_0^{t_0}
\frac{d-1}{d}\frac{\hbar^2}{m_0\ep_0}
 \biggnorm{\frac{\hat\rho}{(\hbar c|\bk|)^2+t^2}}_{L^2(\RR^d)}^{2}
 c_{t_0}^\ell     \,t^2dt   \\
& \geq   \sum_{\ell =0}^\infty (Q(1+\vep))^{2\ell +2} \frac{d}{\pi} \int_0^{t_0}
\frac{d-1}{d}\frac{\hbar^2}{m_0\ep_0}
 \biggnorm{\frac{\hat\rho}{(\hbar c|\bk|)^2+t^2}}_{L^2(\RR^d)}^{2}
 \left(\frac{Q^{-2}}{(1+\vep)^2}\right)^\ell     \,t^2dt   \\
& \geq   \sum_{\ell =0}^\infty (Q(1+\vep))^{2} \frac{d}{\pi} \int_0^{t_0}
\frac{d-1}{d}\frac{\hbar^2}{m_0\ep_0}
 \biggnorm{\frac{\hat\rho}{(\hbar c|\bk|)^2+t^2}}_{L^2(\RR^d)}^{2} \,t^2dt =+\infty.
\end{align}
Hence the radius of convergence of the Maclaurin series of 
$E(q,\boldsymbol{0})$ is equal to $Q$.

In the same way as in the proof of Theorem \ref{thm:conv_rad_gs_single}, 
one can show that the Maclaurin series of 
$\til{\Phi}_\mathrm{g}(q,\boldsymbol{0})$ has a radius of convergence $Q$.
\end{proof}

\subsubsection{\texorpdfstring{The case $\bP\neq\boldsymbol{0}$}{The case P \textne 0}}
We discuss the holomorphy in $q$ of $\til{\Phi}_{\mathrm{g}}(q,\bP)$ and $E(q,\bP)$ for $\bP\neq\boldsymbol{0}$.  
By the unitary operator $\sU$ from Theorem \ref{diagonalization of TDPF}, 
the Hamiltonian is transformed as
\begin{equation}
 \sU H(q,\bP) \sU^* 
   = \dGb(S(q)) - \frac{q}{m_0} 
      \sum_{j=1}^d \PhiS(S(q)^{-1/2} T^{1/2} P_jg_j) + E(q,\boldsymbol{0})+\frac{\bP^2}{2m_0},  \label{eq:H(q,bP)}
\end{equation}
which is of van Hove--Miyatake type, whose properties have been studied in detail 
using Weyl operators (see \cite[Chapter 13]{AraiBook2}).
Moreover, from \cite[Theorem 7.6]{MR4213757}, the lowest energy $E(q,\bP)$ is expressed as 
\begin{equation}
  E(q,\bP) = \frac{\bP^2}{2m_0(1+q^2Q^{-2})} + E(q,\boldsymbol{0}).   \label{eq:E(q,bP)}
\end{equation}
Hence, by Theorem \ref{thm:20261229}, the following statement holds immediately.

\begin{Theorem}
Let $\bP\in\RR^d\setminus \{\boldsymbol{0}\}$. 
The lowest energy $E(q,\bP)$ can be extended as a holomorphic function on the complex domain
\begin{equation}
  \CC \setminus \{ bi \mid b\in\RR, \; |b|\geq Q \}.
\end{equation}
\end{Theorem}

\begin{Theorem}\label{thm:conv_rad_PF_P}
Let $\bP\in\RR^d \setminus \{\boldsymbol{0}\} $. 
The radius of convergence of the Maclaurin series of $E(q,\bP)$ is $Q$.
\end{Theorem}

\begin{proof}
By using Lemma \ref{lem:E(q,0)_int}, the Maclaurin series of $E(q,\bP)$ can be written,
using $c_t$ defined in \eqref{def:c_t}, as
\begin{align}
&  \frac{\bP^2}{2m_0} + \sum_{\ell=0}^\infty (-1)^\ell \left(\frac{q}{Q}\right)^{2\ell+2}
   \bigg\{ - \frac{\bP^2}{2m_0}  \\
& + \frac{d}{\pi} \int_0^\infty
    \frac{d-1}{d}\frac{\hbar^2}{m_0\ep_0}
    \Bignorm{\frac{\hat\rho}{(\hbar c|\bk|)^2+t^2}}_{L^2(\RR^d)}^2
    c_t^\ell Q^{2\ell+2} \,t^2dt \bigg\}.  \label{eq:EPF_div2}
\end{align}
For the series \eqref{eq:EPF_div2} to converge at $q=i(1+\vep)Q$,
it is necessary that
\begin{align}
  \lim_{\ell \to\infty}
   \frac{d}{\pi} \int_0^\infty
   \frac{d-1}{d}\frac{\hbar^2}{m_0\ep_0}
  \Bignorm{\frac{\hat\rho}{(\hbar c|\bk|)^2+t^2}}_{L^2(\RR^d)}^{2}
  c_t^\ell Q^{2\ell +2} \,t^2dt  = \frac{\bP^2}{2m_0} \neq 0.  \label{eq:EPF_div3}
\end{align}
However, for every $t>0$, we have $(c_t Q^2)^\ell \to 0$ as $\ell\to\infty$.
By the dominated convergence theorem, the left-hand side converges to $0$.
Therefore \eqref{eq:EPF_div3} cannot hold.
\end{proof}

As shown in \cite[Theorem 7.6]{MR4213757}, 
 the necessary and sufficient condition for the Hamiltonian $H(q,\bP) \; (\bP\neq \boldsymbol{0})$
 to have a ground state is the infrared regularity condition
\begin{equation}
 |\bk|^{-3/2}\hat\rho\in L^2(\RR^d),
\end{equation}
equivalently $g_j\in\dom(T^{-1})$ for $j=1,\cdots,d$, under which 
$H(q,\bP)$ has a unique ground state
\begin{align}
  \sU^* \exp \Big[ 
  -i \frac{q}{m_0}  \PhiS\big(i S(q)^{-3/2} T^{1/2} \bP \cdot \boldsymbol{g}\big)
  \Big] \Omega.
\end{align}
Note also that $\dom(S(q)^{-3/2})=\dom(T^{-3/2})$ holds\,(\cite[Lemma B1]{MR4213757}).

Recalling that $\sU^* A(f)\sU = B(f)$, the above expression can be written as
\begin{align}
& \exp \Big[ 
  -i \frac{q}{m_0}  \sU^* \PhiS\big(i S(q)^{-3/2} T^{1/2} \bP \cdot \boldsymbol{g} \big) \sU\Big]
   \sU^* \Omega \\
&  = \exp \Big[ 
  -i \frac{q}{m_0} \frac{1}{\sqrt{2}} \ovl{\big( B(i S(q)^{-3/2} T^{1/2} \bP \cdot \boldsymbol{g})+
   B^*(i S(q)^{-3/2} T^{1/2} \bP \cdot \boldsymbol{g})  \big)} \Big]
   \sU^* \Omega\\
& = \exp \Big[ 
  -i \frac{q}{m_0} 
   \PhiS\big( T^{-1/2}S(q)^{1/2} i S(q)^{-3/2} T^{1/2} \bP \cdot \boldsymbol{g}     \big) \Big]
   \sU^* \Omega \\
& = \exp \Big[ 
    -i \frac{q}{m_0} 
    \PhiS\big( iT^{-1/2} S(q)^{-1} T^{1/2} \bP \cdot \boldsymbol{g}     \big) \Big]
    \sU^* \Omega.  \label{GSDPF001}
\end{align}
We set
\begin{align}
 G(q) \coloneqq  \frac{1}{\sqrt{2}} \frac{q}{m_0} T^{-1/2} S(q)^{-1} T^{1/2} \bP \cdot \boldsymbol{g}.
\end{align}
Recalling that $\Phigt(q,\boldsymbol{0})$ is obtained from $\sU^* \Omega$
by removing the normalization constant, the corresponding ground state
of $H(q,\bP)$ is
\begin{equation}
  c(q) \exp\big(-i\sqrt{2}\PhiS(iG(q))\big) \Phigt(q,\boldsymbol{0}),
\end{equation}
where the constant $c(q)$ will be fixed so that its inner product with $\Omega$ equals $1$
for the perturbation expansion.

Using the properties of the coherent vector (\cite[Theorem 5.26]{AraiBook2}) and 
Theorem \ref{thm:1_2_inner} with 
$\Phigt(q,\boldsymbol{0}) = e^{-\frac{1}{2}\Delta^*(K_1(q))}\Omega$,
we have
\begin{align}
  \biginner{\Omega}{\exp(-i\sqrt{2}\PhiS(iG(q))) \Phigt(q,\boldsymbol{0})}
& = \biginner{\exp(i\sqrt{2}\PhiS(iG(q))) \Omega}{\Phigt(q,\boldsymbol{0})} \\
& = e^{-\norm{G(q)}^2/2} \biginner{ \exp(-A^*(G(q))) \Omega}{\Phigt(q,\boldsymbol{0})} \\
& = e^{-\norm{G(q)}^2/2} e^{-\inner{G(q)}{K_1(q)G(q)}/2},
\end{align}
where we used the fact that $JG(q)=G(q)$. 
Then, the vacuum-normalized ground state of $H(q,\bP)$ is
\begin{align}
  \Phigt(q,\bP) 
  &\coloneqq  e^{\norm{G(q)}^2/2} e^{\inner{G(q)}{K_1(q)G(q)}/2} 
      \exp\big( A^*(G(q))-A(G(q)) \big) e^{-\frac{1}{2}\Delta^*(K_1(q))} \Omega.
\end{align}

\begin{Theorem}\label{thm:dpf_holo}
Suppose that $\bP\neq \boldsymbol{0}$ and $|\bk|^{-3/2}\hat\rho\in L^2(\RR^d)$.
 Then the ground state $\til{\Phi}_\mathrm{g}(q,\bP)$ of $H(q,\bP)$ can be extended
 as a holomorphic function on the complex domain $U$ defined by \eqref{def of UTDPF}.
In particular, the Maclaurin series for $\til{\Phi}_\mathrm{g}(q,\bP)$ converges absolutely for $|q| < Q$.
\end{Theorem}

\begin{Example}
We consider the same physical situation as in Example \ref{example:Q2}.
Taking mass renormalization into account, we regard $m_0$ as the bare mass.
That is, $m_0 = m_\mathrm{bare}$ is determined so as to satisfy the relation
\eqref{def:mbare} between $m_\mathrm{bare}$ and the observed electron mass $m$.
Moreover, we set $d=3$ and define the coupling function $\hat\rho$ by
\eqref{std uv}.
In this case, $m_0 = m_\mathrm{bare}$ is given by \eqref{eq:bare_mass_relation}.
Consequently, $Q$ is given by \eqref{eq:Q_d3_bare}, as in Example \ref{example:Q2}.
Since $e < Q$ holds, it follows from Theorem \ref{thm:conv_rad_PF_P} that
the Maclaurin series for $E(q,\bP)$ converges at $q=-e$.

In order to study the holomorphy of the ground state in the case $\bP \neq \boldsymbol{0}$,
it is necessary to impose an infrared regularity condition.
Let $k_\mathrm{min}>0$ and define
\begin{align}
  \hat\rho(\bk) \coloneqq  (2\pi)^{-3/2}\chi(k_\mathrm{min}<|\bk|<k_\mathrm{c}),
\end{align}
where we assume $k_\mathrm{min}<k_\mathrm{c}$.
In this case, from \eqref{def:mbare}, the corresponding bare mass is given by
\begin{align}
  m_\mathrm{bare}
  = m\left( 1 - \frac{4\alpha}{3\pi} \left( 1-\frac{k_\mathrm{min}}{k_\mathrm{c}} \right) \right).
\end{align}
Moreover, from the definition \eqref{def of Q}, we obtain
\begin{align}
  Q = \bigg\{ 
      \frac{3\pi}{4\alpha} \bigg( 1-\frac{k_\mathrm{min}}{k_\mathrm{c}} \bigg)^{-1}  -1 
      \bigg\}^{1/2} e.
\end{align}
Since $Q$ is increasing in $k_\mathrm{min}>0$, it follows that
\begin{align}
  Q > \bigg(\frac{3\pi}{4\alpha}-1\bigg)^{1/2} e > 17.941\,e.
\end{align}
Therefore, it follows from Theorem \ref{thm:dpf_holo} that, 
for any infrared cutoff $k_\mathrm{min}$, the Maclaurin series of the ground state
 $\Phigt(q,\bP)$ of $H(q,\bP)$
converges at $q=-e$.
\end{Example}

\medskip
We now turn to the proof of Theorem \ref{thm:dpf_holo}.
To this end, we prepare several lemmas.

\begin{Lemma}\label{lem:G_conti}
Assume that $|\bk|^{-3/2}\hat\rho\in L^2(\RR^d)$.
Then $G(q)$ can be analytically continued to $U$.
\end{Lemma}
\begin{proof}
Since $S(q)^{-2}=T^{-1} (\one+ \ovl{T^{-1}W(q)T^{-1}})^{-1} T^{-1} $, we can write $G(q)$ as
\begin{align}
 G(q)
 & = \frac{q}{\sqrt{2} m_0}
     T^{-1/2} S(q) T^{-1} (\one+ \ovl{T^{-1}W(q)T^{-1}})^{-1} 
     T^{-1/2} \bP \cdot \boldsymbol{g} \\
 & = \frac{q}{\sqrt{2} m_0} (T^{-1/2} S(q) T^{-1/2}) T^{-1/2}
     (\one+ \ovl{T^{-1}W(q)T^{-1}})^{-1} T^{-1/2} \bP \cdot \boldsymbol{g}.
\end{align}
By Proposition~\ref{prop:tdpf0cs} and the proof of Theorem \ref{thm:holo_K1},
$K_4(q)$ admits an analytic continuation to $U$.
Hence, $\ovl{T^{-1/2} S(q)T^{-1/2}} = 2K_4(q)+\one$ also extends to a holomorphic 
function on $U$. Moreover,
\begin{align}
 (\one+\ovl{T^{-1}W(q)T^{-1}})^{-1} 
    = \one - \frac{q^2}{m_0(1+q^2Q^{-2})}\sum_{j=1}^d \ketbra{T^{-1/2}g_j}{T^{-1/2}g_j}
\end{align}
extends to a holomorphic function except at $z=\pm iQ$.
Therefore, $G(q)$ can be analytically continued to $U$ as
\begin{align}
  G(z) 
 \coloneqq  \frac{z}{\sqrt{2} m_0} (2K_4(z)+\one) 
   \bigg( T^{-1} - \frac{z^2}{m_0(1+z^2Q^{-2})}\sum_{j=1}^d \ketbra{T^{-1}g_j}{T^{-1}g_j} \bigg)  \bP \cdot \boldsymbol{g}.
\end{align}
\end{proof}

\begin{Lemma}\label{lem:holo_constant}
Assume that $|\bk|^{-3/2}\hat\rho\in L^2(\RR^d)$.
Then $\inner{G(q)}{K_1(q)G(q)}$ extends to the holomorphic function 
$\inner{G(\bar{z})}{K_1(z)G(z)}$ on $U$.
\end{Lemma}

\begin{proof}
By Proposition~\ref{prop:tdpf0cs} and Theorem~\ref{thm:holo_K1}, 
$K_1(q)$ extends to a holomorphic function on $U$.
Moreover, by Lemma~\ref{lem:G_conti}, $G(q)$ can also be analytically continued to 
$U$.
\end{proof}

We extend the ground state $\Phigt(q,\bP)$ to $z\in U$ by
\begin{align}
  \Phigt(z,\bP) 
  \coloneqq  e^{\inner{G(\bar{z})}{G(z)}/2} e^{\inner{G(\bar{z})}{K_1(z)G(z)}/2} 
    \exp\! \big(A^*(G(z))-A(G(\bar{z}))\big) e^{-\frac{1}{2}\Delta^*(K_1(z))} \Omega.
  \label{eq:tilphigzP}
\end{align}
Although the operator $\exp(A^*(G(z))-A(G(\bar{z})))$ is unbounded in general,
the above expression is well defined by Theorem~\ref{thm:analytic_vectors}.

For the present model, we have $T=JTJ$ and $Jg_j=g_j$, and hence
\begin{equation}
T^2+W(z)=(T^2+W(z))^\mathrm{T}.
\end{equation}
Using the integral representation, this implies $S(z)=S(z)^\mathrm{T}$,
and therefore $K_1(z)=K_1(z)^\mathrm{T}$.
Moreover, by Lemma~\ref{Lem:K1norm}, we have $\norm{K_1(z)}<1$ for all $z\in U$.

\begin{proof}[Proof of Theorem \ref{thm:dpf_holo}]
By Lemma \ref{lem:holo_constant}, the scalar factor in \eqref{eq:tilphigzP} 
is holomorphic. Hence, it suffices to prove that
\begin{align}
  \Psi(z) \coloneqq  \exp\!\big( A^*(G(z)) - A(G(\bar{z}))  \big)\,
             e^{-\frac{1}{2}\Delta^*(K_1(z))}\Omega
\end{align}
is holomorphic on $U$. First, we consider 
\begin{align}
  \Psi_n(z) \coloneqq  \big( A^*(G(z)) - A(G(\bar{z}))  \big)^n 
               e^{-\frac{1}{2}\Delta^*(K_1(z))}\Omega, 
  \qquad (n=0,1,2,\cdots),
\end{align}
and show that it is holomorphic. By Theorem \ref{main2},
$\Psi_0(z) = e^{-\frac{1}{2}\Delta^*(K_1(z))}\Omega$ is holomorphic.
If the bounded operator 
\begin{align}
  \big( A^*(G(z)) - A(G(\bar{z})) \big)^n (\Nb+1)^{-n},  \label{eq:AznN}
\end{align}
and the vector 
\begin{align}
  \Xi(z) \coloneqq  (\Nb+1)^n \Psi_0(z)
\end{align}
are holomorphic, then $\Psi_n(z)$ is holomorphic as well.
Here $\Xi(z)$ is well defined by Theorem~\ref{thm:GS-prop2}.

The function $\Psi_0(z)$ satisfies
\begin{align}
 \frac{d}{dz}\Psi_0(z)  &= -\frac{1}{2}\Delta^*(K_1'(z)) \Psi_0(z), \\
 \frac{d^2}{dz^2}\Psi_0(z) &= -\frac{1}{2}\Delta^*(K_1''(z)) \Psi_0(z)
     +\frac{1}{4}\big(\Delta^*(K_1'(z))\big)^2 \Psi_0(z).
\end{align}
By Lemma~\ref{lem:holo_AN}, the operator \eqref{eq:AznN} is holomorphic on $U$.

We now show that $\Xi(z)$ is holomorphic by applying Lemma~\ref{lem:criterion-holomorphy}.
Define
\begin{align}
  \Xi_1(z) &\coloneqq -\frac{1}{2}(\Nb+1)^n \Delta^*(K_1'(z)) \Psi_0(z), \\
  \Xi_2(z) &\coloneqq -\frac{1}{2}(\Nb+1)^n \Delta^*(K_1''(z)) \Psi_0(z)
     +\frac{1}{4}(\Nb+1)^n \big( \Delta^*(K_1'(z)) \big)^2 \Psi_0(z).
\end{align}
From the definition of $\Delta^*(\cdot)$, we have
\begin{align}
 & \dom\!\big((\Nb+1)^n\Delta^*(K_1'(z))\big) \subset \dom(\Nb^{\,n+1}), \\ 
 & \dom\!\big((\Nb+1)^n\Delta^*(K_1''(z))\big) \subset \dom(\Nb^{\,n+1}), \\
 & \dom\!\big(((\Nb+1)^n\Delta^*(K_1'(z)))^2\big) \subset \dom(\Nb^{\,n+2}),
\end{align}
and hence $\Xi_1(z)$ and $\Xi_2(z)$ are well-defined.

Note that, for any $\Phi \in \Ffin(\sH)$, the function $\inner{\Phi}{\Xi(z)}$ 
is holomorphic on $U$, and
\begin{align}
  \frac{d}{dz} \inner{\Phi}{\Xi(z)} = \inner{\Phi}{\Xi_1(z)}, \qquad 
  \frac{d^2}{dz^2} \inner{\Phi}{\Xi(z)} =  \inner{\Phi}{\Xi_2(z)}.
\end{align}
We show that $\norm{\Xi_2(z)}$ is locally bounded on $U$.
By \eqref{Delta prop2}, we obtain
\begin{align}
  \norm{\Xi_2(z)}
& \leq  \frac{1}{2}\norm{K_1''(z)}_2 \,
         \norm{(\Nb+2)(\Nb+3)^n \Psi_0(z)} \\
& \quad +\frac{1}{4}\norm{K_1'(z)}_2 \,
         \norm{(\Nb+2)\Delta^*(K_1'(z))(\Nb+5)^n \Psi_0(z)} \\
& \leq  \frac{1}{2}\norm{K_1''(z)}_2 \,
         \norm{(\Nb+2)(\Nb+3)^n \Psi_0(z)} \\
& \quad +\frac{1}{4}\norm{K_1'(z)}_2^2 \,
         \norm{(\Nb+2)(\Nb+4)(\Nb+5)^n \Psi_0(z)}.
\end{align}
By Theorem~\ref{thm:GS-prop2}, it follows that $\Xi_2(z)$ is bounded 
on compact subsets of $U$. 
Hence, by Lemma~\ref{lem:criterion-holomorphy}, 
$\Xi(z)$ and therefore $\Psi_n(z)$ are holomorphic in $U$.

Finally, it remains to show that
\begin{align}
  \sum_{n=0}^\infty \frac{\Psi_n(z)}{n!} = \Psi(z)
\end{align}
converges uniformly on compact subsets of $U$. 
Let $V \subset U$ be a nonempty compact set.
Since $\norm{K_1(z)}$, $\norm{G(\bar{z})}$, and $\norm{G(z)}$ are continuous on $V$, we can set
\begin{align}
  M \coloneqq  2 \sup_{z \in V} (1 - \norm{K_1(z)}^2)^{-1/2}
       \big(\norm{G(\bar{z})} + \norm{G(z)}\big) < \infty.
\end{align}
By Theorem~\ref{thm:analytic_vectors} with $\gamma \coloneqq  1 - K_1(z)K_1(z)^*$, we obtain
\begin{align}
  \norm{\Psi_n(z)}
& \leq 2^n \sqrt{n!}\,
       \big(\norm{\gamma^{-1/2}G(\bar{z})} + \norm{\gamma^{-1/2}G(z)}\big)^n \norm{\Psi_0(z)} \\
& \leq 2^n \sqrt{n!}\,
       \norm{(\one - K_1(z)K_1^*(z))^{-1/2}}^n
       \big(\norm{G(\bar{z})} + \norm{G(z)}\big)^n \norm{\Psi_0(z)} \\
& \leq 2^n \sqrt{n!}\,
       (1 - \norm{K_1(z)}^2)^{-n/2}
       \big(\norm{G(\bar{z})} + \norm{G(z)}\big)^n \norm{\Psi_0(z)} \\
& \leq \sqrt{n!}\, M^n \sup_{z \in V}\norm{\Psi_0(z)}.
\end{align}
Therefore, by the Weierstrass $M$-test, the series 
$\sum_{n=0}^\infty \Psi_n(z)/n!$ converges to $\Psi(z)$ uniformly on compact subsets of $U$.
This shows that $\Psi(z)$ is holomorphic on $U$.
\end{proof}

\begin{Theorem}
Suppose that $\bP\neq \boldsymbol{0}$ and $|\bk|^{-3/2}\hat\rho\in L^2(\RR^d)$.
The radius of convergence of the Maclaurin series for the ground state $\til{\Phi}_\mathrm{g}(q,\bP)$ is $Q$.
\end{Theorem}
\begin{proof}
It is shown in Theorem \ref{thm:conv_rad_PF_P} that the radius of convergence of the Maclaurin series 
for the ground state energy $E(q,\bP)$ is $Q$. 
Therefore, in the same way as in the proof of Theorem \ref{thm:conv_rad_gs_single}, 
one can show that the Maclaurin series of $\Phigt(q,\bP)$ also has a radius of convergence $Q$.
\end{proof}

\appendix
\section{Heinz type inequalities for m-accretive operators}\label{AppenA}
The Heinz inequality for bounded self-adjoint operators states that if $0\leq A\leq B$, then $A^p \leq B^p$ holds for $0<p<1$.
This can be generalized to the case where $B$ is m-accretive.
We note that for an m-accretive operator $B$ and $0<p<1$, its fractional power $B^p$ is defined as the 
closure of the operator given by
\begin{align}
  B^p \psi  & = \frac{\sin \pi p}{\pi}\int_0^\infty \lambda^{p-1} B(B+\lambda)^{-1} \psi  \, d\lambda \\
      & = \frac{\sin \pi p}{\pi}\int_0^\infty  \lambda^{p-1} 
          \left(1-\lambda(B+\lambda)^{-1}\right)  \psi \, d\lambda,  \qquad \psi\in \dom(B)\label{def:B^p}
\end{align}
(see \cite[p.~286, eq.~(3.53)]{MR1335452}, where the integral converges absolutely in the strong sense).

\begin{Theorem}\label{thm:g-Heinz1}
Let $\sH$ be a Hilbert space, $A$ be a non-negative self-adjoint operator, and 
$B$ be an m-accretive operator on $\sH$. 
Suppose that $B=B_0 + iB_1$ for some self-adjoint operators $B_0$ and $B_1$ 
such that $\dom(B_0)\subset \dom(B_1)$.
Assume that $\dom(B_0)\subset \dom(A)$ and
\begin{align}
  \inner{u}{Au} \leq \inner{u}{B_0u},  \qquad u \in \dom(B_0).
\end{align}
Then, for all $0<p<1$,
\begin{align}
  \inner{u}{A^pu} \leq \Re \inner{u}{B^pu},
  \qquad u \in \dom(B).  \label{gen:Heinz}
\end{align}
\end{Theorem}
\begin{proof}
 Since $B$ is closed and m-accretive, $B^p$ can be defined by \eqref{def:B^p}.
By using the Heinz inequality (e.g., \cite[Proposition 10.14]{MR2953553}), we have
\begin{align}
  \inner{u}{A^pu} \leq \inner{u}{(B_0)^pu}, \qquad u\in\dom(B_0). \label{ApB0p}
\end{align}
From the assumptions, we have $\dom(B_0)\subset \dom(B_1)$, and therefore, $|B_1|^{1/2}$ is
 $B_0^{1/2}$-bounded. Hence $C\coloneqq  (B_0+\lambda)^{-1/2}B_1 (B_0+\lambda)^{-1/2}$ is bounded.

For $\lambda>0$, we have the operator equality
\begin{align}
  B+\lambda    =  (B_0+\lambda)^{1/2}(1+ iC) (B_0+\lambda)^{1/2}.
\end{align}
Thus we have
\begin{align}
  (B+\lambda)^{-1} 
  = (B_0+\lambda)^{-1/2}(1+ i\bar{C})^{-1} (B_0+\lambda)^{-1/2}.
\end{align}
Since $\bar{C}$ is self-adjoint, we have
\begin{align}
  \Re( (1+i\bar{C})^{-1}) = (1+\bar{C}^2)^{-1} \leq 1,
\end{align}
and hence
\begin{align}
  \Re ((B+\lambda)^{-1}) \leq  (B_0+\lambda)^{-1}.
\end{align}
For $u\in\dom(B)=\dom(B_0)$, by \eqref{def:B^p}, we have
\begin{align}
  \Re \inner{u}{B^pu} 
& = \frac{\sin \pi p}{\pi}\int_0^\infty \lambda^{p-1} 
        \Re\inner{u}{\left(1-\lambda(B+\lambda)^{-1}\right)u} \, d\lambda \\
& \geq \frac{\sin \pi p}{\pi}\int_0^\infty \lambda^{p-1} 
        \inner{u}{\left(1-\lambda(B_0+\lambda)^{-1}\right)u} \, d\lambda \\
& = \inner{u}{(B_0)^pu}.
\end{align}
This inequality and \eqref{ApB0p} imply \eqref{gen:Heinz}.
\end{proof}

\section{Criterion for Holomorphy}\label{AppenB}
\begin{Lemma}\label{lem:criterion-holomorphy}
Let $\sH$ be a Hilbert space, and let $U \subset \CC$ be a nonempty open set.  
Consider a map $f:U \to \sH$.  
Suppose $D \subset \sH$ is dense, and that for every $\Phi \in D$ the function 
$z \mapsto \inner{\Phi}{f(z)}$ is holomorphic on $U$.  
Assume that for each $z \in U$ there exist vectors $\Xi_1(z), \Xi_2(z) \in \sH$ such that
\begin{align}
  \frac{d}{dz} \inner{\Phi}{f(z)} &= \inner{\Phi}{\Xi_1(z)}, \qquad
  \frac{d^2}{dz^2} \inner{\Phi}{f(z)} = \inner{\Phi}{\Xi_2(z)}
\end{align}
for all $\Phi \in D$.  
Suppose further that $\norm{\Xi_2(z)}$ is bounded on compact subsets of $U$. 
Then $f(z)$ is holomorphic, and $\tfrac{d}{dz} f(z) = \Xi_1(z)$.
\end{Lemma}

\begin{proof}
For any $\Phi\in D$ and for all $h\in \CC$ with $0<|h|<(1/2)\mathrm{dist}(z,U^\mathrm{c})$, we have
\begin{align}
 \inner{\Phi}{ \frac{f(z+h)-f(z)}{h} - \Xi_1(z)}
 &= \int_0^1 \inner{\Phi}{\Xi_1(z+ht)-\Xi_1(z)} \, dt \\
 &= \int_0^1 \int_0^1\inner{\Phi}{\Xi_2(z+hst)}ht \, ds\,dt .
\end{align}
Since $\norm{\Xi_2(z)}$ is bounded on every compact neighborhood of $z \in U$, we obtain
\begin{align}
  \Bigl\| \frac{f(z+h)-f(z)}{h} - \Xi_1(z) \Bigr\|
&= \sup_{\Phi \in D, \norm{\Phi}=1 } \bigg|
     \int_0^1\int_0^1 \inner{\Phi}{\Xi_2(z+hst)}ht \, ds\,dt \bigg| \\
& \leq |h| \sup_{0\leq s,t\leq 1} \norm{\Xi_2(z+hst)} \\
& \to 0 \qquad (h\to 0).
\end{align}
This completes the proof.
\end{proof}

\section{Two-Particle Creation Operators and Bogoliubov Transformations}
\label{AppenC}
In the construction of Bogoliubov transformations,
creation and annihilation operators associated with two-particle states
play an essential role.
In particular, given a Hilbert--Schmidt operator,
one can associate a two-particle creation operator,
whose exponential naturally appears in the description of ground states
of bosonic quadratic Hamiltonians.
We refer to vectors obtained by applying such exponentials
to the vacuum vector $\Omega$ as squeezed states.

This appendix collects several general facts concerning two-particle creation operators.
Our main focus is on the analytic properties of squeezed states.
We also present an explicit construction of a Bogoliubov transformation
associated with a given squeezed state.
More precisely, we show that such a vector can always be realized as the vacuum
for a suitably constructed Bogoliubov transformation.
This provides a general and model-independent correspondence
between Hilbert--Schmidt operators and Bogoliubov transformations,
which may be of independent interest.
Additional results needed in the main text are also included for completeness.

Let $\cS_2(\sH)$ be the set of Hilbert--Schmidt operators on $\sH$.
The space $\cS_2(\sH)$ with the inner product $\inner{A}{B} \coloneqq  \tr(A^*B)$~($A,B\in\cS_2(\sH)$)
is a Hilbert space. 
The associated norm is $\norm{A}_2\coloneqq  \tr(A^*A)^{1/2}$.
For a fixed conjugation $J$ on $\sH$\, (i.e., an anti-linear isometric involution),
there exists a unique unitary operator $\Theta:\cS_2(\sH)\to \sH\tensor \sH$ such that 
\begin{align}
  \Theta (\ketbra{f_1}{Jf_2}) = f_1\tensor f_2,
 \qquad f_1,f_2 \in \sH.
\end{align}
We note that for $K\in \cS_2(\sH)$, the condition $K=JK^*J$ is equivalent 
to $\Theta(K)\in \stensor^2 \sH$.
In other words, any symmetric two-particle state can be identified with a Hilbert--Schmidt operator $K$ satisfying $K=K^\mathrm{T}$ with respect to the conjugation $J$.

We introduce the creation operators for two-particle states. 
For $K \in \cS_2(\sH)$, we define
\begin{align}
  & \dom(\Delta^*(K)) \coloneqq  \left\{ \Psi = (\Psi^{(n)})_{n=0}^\infty \, \bigg| \, \sum_{n\geq 2} n(n-1) \norm{S_n(\Theta(K)\tensor \Psi^{(n-2)})}^2 < \infty\right\} \\
 & (\Delta^*(K) \Psi)^{(n)} \coloneqq  
  \begin{cases}
    0  & \qquad n=0,1 \\
   \sqrt{n(n-1)}S_n(\Theta(K)\tensor \Psi^{(n-2)}), & \qquad n\geq 2,
  \end{cases}
\end{align}
where $S_n$ denotes the symmetrization operator.
The two-particle annihilation operator is defined by $\Delta(K)\coloneqq  (\Delta^*(K))^*$.
One can easily verify the following facts:
\begin{itemize}
\item $\Delta^*(K)$ is a densely defined closed operator on $\Fb(\sH)$.
\item The domain of the number operator $\Nb \coloneqq  \dGb(\one)$ satisfies
  \begin{align}
    \dom(\Nb) \subset \dom(\Delta(K)) \cap \dom(\Delta^*(K)).
  \end{align}
  Moreover, for every $\Psi\in\dom(\Nb)$, we have
  \begin{align}
    \norm{\Delta^*(K)\Psi} \leq \norm{K}_2  \norm{\sqrt{(\Nb+1)(\Nb+2)}\Psi}. \label{Delta prop2}
  \end{align}
\item For any $m\in \NN$ and $\Psi \in \dom(\Nb^m)$, we have $\Delta^*(K)\Psi \in \dom(\Nb^{m-1})$.
\item For $u,v\in \sH$, we have $\Delta^*(\ketbra{u}{Jv}) = A^*(u)A^*(v)$ on $\dom(\Nb)$.
\item For every $K \in \cS_2(\sH)$, we have $\Delta^*(K) = \Delta^*((K+K^\mathrm{T})/2)$, 
which shows that $\Delta^*(K)$ corresponds to the creation operator of the 
symmetric two-particle state $\Theta((K+K^\mathrm{T})/2) = S_2\Theta(K)$.
\end{itemize}

Next, we review some basic known results on Bogoliubov transformations.
The main references are \cite{Rui78,EB}.
Note that there are differences in notation between these references and this paper.

Let $\Ffin(\sH)$ be the subspace spanned by the Fock vacuum $\Omega$ and vectors
of the form
$A^*(f_1)\cdots A^*(f_n)\Omega$ $(f_j \in \sH, j=1,\cdots,n; n\in\NN)$.
For $L \in \cB(\sH)$, we can define
\begin{align}
  \wick{e^{-\dGb(L)}} \Psi \coloneqq  \slim_{N\to\infty}\sum_{n=0}^N \frac{1}{n!} \wick{(-\dGb(L))^n} \Psi
\end{align}
for $\Psi \in \Ffin(\sH)$, where the symbol $\wick{\text{--}}$ denotes the 
Wick ordering\,(see \cite{Rui78,EB}).

The set of pairs $(X,Y)\in \cB(\sH)\times \cB(\sH)$ of bounded operators that satisfy
\begin{align}
  X^*X - Y^*Y = 1, \qquad X^*JYJ - Y^*JXJ = 0, \\
  XX^*-JYY^*J = 1, \qquad XY^* - JYX^*J=0,     \label{symp_conditions}
\end{align}
is denoted by $\mathfrak{Sp}$.
For $(X,Y)\in \mathfrak{Sp}$, we set
\begin{align}
  K_1\coloneqq  JYX^{-1}J, \qquad K_2\coloneqq  1-(X^{-1})^*, \qquad 
  K_3\coloneqq  -JX^{-1}JY.  \label{def:Ks}
\end{align}
Then, it holds that 
\begin{align}
  K_j = K_j^\mathrm{T}, \qquad \norm{K_j}<1, \qquad (j=1,3).
\end{align}

In the following, we assume that $Y\in \cS_2(\sH)$. 
Then, $K_1, K_3 \in \cS_2(\sH)$. We set 
\begin{align}
 B(f) \coloneqq  \ovl{A(Xf)+A^*(JYf)},  \qquad f\in \sH.  \label{def:B(f)}
\end{align}
The operators $\{B(f),B^*(f) \mid f\in\sH\}$ provide a representation of the canonical commutation relations (CCR).
Solving \eqref{def:B(f)} for $A(f)$, we have
\begin{align}
  A(f) = \ovl{B(X^*f)-B^*(Y^*Jf)}.  \label{eq:A(f)B(f)}
\end{align}
It is known that there exists a unitary operator $\sU$ on $\Fb(\sH)$ such that 
\begin{align}
  \sU B(f) \sU^* = A(f),\qquad f\in \sH.
\end{align}
This operator $\sU$ is given by 
\begin{align}
  \sU^* = \det(1-K_1^*K_1)^{1/4} e^{-\frac{1}{2}\Delta^*(K_1)} 
      \wick{ e^{-\dGb(K_2)} } e^{-\frac{1}{2} \Delta(K_3^*)}   \label{defU*}
\end{align}
on $\Ffin(\sH)$, up to complex multiples of absolute value one\,(see \cite[Proposition 2.5]{EB} and \cite{Rui78}).
Since
\begin{align}
      \wick{ e^{-\dGb(K_2)} } e^{-\frac{1}{2} \Delta(K_3^*)}\Omega = \Omega,
\end{align}
one has
\begin{align}
  \sU^* \Omega = \det(1-K_1^*K_1)^{1/4} e^{-\frac{1}{2}\Delta^*(K_1)} \Omega.
\end{align}

\begin{Theorem}\label{thm:squeezed-state}
Let $K\in \cS_2(\sH)$ satisfy the condition
\begin{align}
  \norm{K}<1.
\end{align}
Then, 
\begin{align}
  e^{-\frac{1}{2}\Delta^*(K)}\Omega 
 \coloneqq  \slim_{N\to\infty} \, \sum_{n=0}^N \frac{(-1)^n}{2^nn!} (\Delta^*(K))^n \Omega \label{24011221300}
\end{align}
converges, and the following bound
\begin{align}
  \norm{e^{-\frac{1}{2}\Delta^*(K)}\Omega}^2
  \leq e^{\norm{K}_2 ^2/(2-2\norm{K}^2)} \label{240527}
\end{align}
holds.
\end{Theorem}

\begin{proof}
First we consider the case of $K=K^\mathrm{T}$. 
By \cite[eq.~(4.36)]{Rui78}, it is proven that the series
\eqref{24011221300} converges, and 
\begin{align}
  \norm{e^{-\frac{1}{2}\Delta^*(K)}\Omega}^2 = \frac{1}{(\det(1-K^*K))^{1/2}}.  \label{norm of gs1}
\end{align}
holds (see also \cite{EB}), where we used the fact that $K=K^\mathrm{T}$. 
Let $\mu_1\geq \mu_2 \geq \cdots$ be the eigenvalues of $K^*K$.
Then we have $\mu_1 = \norm{K}^2 < 1$ and $\norm{K}_2 ^2= \sum_j \mu_j<\infty$ by assumption.
By using the inequality $(1-\mu_j)^{-1} \leq e^{\mu_j/(1-\mu_1)}$, we have the bound
\begin{align}
  (\det(1-K^*K))^{-1/2} = \Big(\prod_j (1-\mu_j)\Big)^{-1/2} \leq e^{\norm{K}_2^2/(2-2\norm{K}^2)}.
\end{align}

Next, we consider the case where $K$ is not necessarily equal to $K^\mathrm{T}$.
Let $L=(K+K^\mathrm{T})/2$. Then $L$ is Hilbert--Schmidt, $L=L^\mathrm{T}$, $\norm{L}<1$
and $\Delta^*(K) = \Delta^*(L)$. Thus the above first case implies that \eqref{24011221300} converges and 
\begin{align}
  \norm{e^{-\frac{1}{2}\Delta^*(K)}\Omega}^2
  \leq e^{\norm{L}_2^2/(2-2\norm{L}^2)}.
\end{align}
Since 
\begin{align}
  \norm{L}_2  
  \leq \frac{1}{2}(\norm{K}_2  +\norm{K^\mathrm{T}}_2 )
  = \norm{K}_2,
\end{align}
the bound \eqref{240527} holds.
\end{proof}

\begin{Theorem}\label{thm:GS-prop2}
Assume the same conditions as in Theorem \ref{thm:squeezed-state}.
Then, for any $k \in \NN$, the vector $e^{-\frac{1}{2}\Delta^*(K)} \Omega$ belongs to $ \dom(\Nb^k) $.
 Furthermore, for any $1 < \nu < \norm{K}^{-2}$, the following inequality holds:
 \begin{align}
\norm{ \Nb^k e^{-\frac{1}{2}\Delta^*(K)} \Omega}^2
\leq \left( \frac{4k}{e\log \nu}\right)^{2k} 
\exp\left( \frac{\nu}{2(1-\nu\norm{K}^2)} \norm{K}_2 ^2 \right). \label{ineq:GS-prop2}
 \end{align}
\end{Theorem}

\begin{proof}
First we consider the case of $K=K^\mathrm{T}$.
Since $(\Delta^*(K))^m\Omega$ is a $2m$-particle state, we have
\begin{align}
  \Nb^k e^{-\frac{1}{2}\Delta^*(K)} \Omega 
  = \sum_{m=0}^\infty \frac{(-1)^m}{2^mm!}(2m)^k (\Delta^*(K))^m\Omega
  \in \prod_{n=0}^\infty \left( \Tensor_\mathrm{s}^n \sH \right).
\end{align}
Therefore, 
\begin{align}
 \norm{ \Nb^k e^{-\frac{1}{2}\Delta^*(K)} \Omega}^2
 = \sum_{m=0}^\infty (2m)^{2k} a_m  \in [0,\infty],  \label{nkgs1}
\end{align}
where
\begin{align}
  a_m = \frac{ \norm{(\Delta^*(K))^m\Omega}^2}{ (2^m m!)^2 }.
\end{align}
Let $F(z) \coloneqq  \sum_{m=0}^\infty a_m z^m$. Then, similarly to 
\eqref{norm of gs1}, for $|z|<\norm{K}^{-2}$, we have
\begin{align}
   F(z) = \prod_{j=1}^\infty (1-z \mu_j)^{-1/2},
\end{align}
which converges absolutely.

On the other hand, by simple calculus computations, we can show that
\begin{align}
  (2m)^{2k} \leq C_{k,\nu} \nu^m, \qquad m=0,1,2,\cdots
\end{align}
where $C_{k,\nu} \coloneqq  (4k/(e\log\nu))^{2k}$. Therefore, 
\begin{align}
  \eqref{nkgs1} \leq C_{k,\nu} \sum_{m=0}^\infty a_m \nu^m = C_{k,\nu} F(\nu).
\end{align}
Using the inequality $(1-\nu\mu_j)^{-1}\leq \exp(\nu\mu_j/(1-\nu\mu_1))$ for the 
product representation of $F(\nu)$, we obtain
\begin{align}
  F(\nu) \leq \exp\left( \frac{\nu}{2(1-\nu\mu_1)}(\mu_1+\mu_2+\cdots)\right) 
             = \exp\left( \frac{\nu}{2(1-\nu\mu_1)} \norm{K}_2 ^2\right).
\end{align}
Thus, the inequality \eqref{ineq:GS-prop2} holds.

Next, we consider the case where $K$ is not necessarily equal to $K^\mathrm{T}$.
Let $L=(K+K^\mathrm{T})/2$. Then, for any $1 < \nu < \norm{K}^{-2}\leq \norm{L}^{-2}$, we have
 \begin{align}
\norm{ \Nb^k e^{-\frac{1}{2}\Delta^*(K)} \Omega}^2
& \leq C_{k,\nu} \exp\left( \frac{\nu}{2(1-\nu\norm{L}^2)} \norm{L}_2 ^2 \right) \\
& \leq  C_{k,\nu} \exp\left( \frac{\nu}{2(1-\nu\norm{K}^2)} \norm{K}_2 ^2 \right).
 \end{align}
This completes the proof.
\end{proof}

\begin{Theorem}\label{gs:Holo}
Let $U$ be an open set in $\CC$.
Assume that $K: U \to \cS_2(\sH)$ is holomorphic, and $\norm{K(z)}<1$ for all $z\in U$.
Then the map
\begin{align}
 U \to \Fb(\sH), \qquad 
    z \mapsto e^{-\frac{1}{2}\Delta^*(K(z))}\Omega
\end{align}
is holomorphic.
\end{Theorem}

\begin{proof}
For each $z\in U$, we set $\Psi(z) \coloneqq  e^{-\frac{1}{2}\Delta^*(K(z))}\Omega$.
We show that $\Psi(z)$ is holomorphic by applying Lemma~\ref{lem:criterion-holomorphy}.

For each $n=1,2,\cdots$, the function $(\Delta^*(K(z)))^n\Omega$ is holomorphic on $U$, 
and 
\begin{align}
  \frac{d}{dz}(\Delta^*(K(z)))^n\Omega
  = n\Delta^*(K'(z)) (\Delta^*(K(z)))^{n-1}\Omega.
\end{align}
Define
\begin{align}
& \Xi_1(z) \coloneqq  -\frac{1}{2}\Delta^*(K'(z)) \Psi(z) \\
& \Xi_2(z) \coloneqq  -\frac{1}{2}\Delta^*(K''(z))  \Psi(z)
            +\frac{1}{4}\Delta^*(K'(z)) \Delta^*(K'(z)) \Psi(z).
\end{align}
These are well-defined. 

Indeed, by Theorem~\ref{thm:GS-prop2}, $\Psi(z)\in\dom(\Nb^2)$.
It follows that $\Delta^*(K'(z))\Psi(z)\in\dom(\Nb)$.
Since $\dom(\Nb)\subset\dom(\Delta(K''(z)))$,
the vector $\Delta^*(K'(z))\Psi(z)$ belongs to $\dom(\Delta(K''(z)))$,
and hence $\Xi_2(z)$ is well-defined.

Since the $n$-particle component of $\Psi(z)$ is holomorphic,
for any $\Phi\in\Ffin(\sH)$, $\inner{\Phi}{\Psi(z)}$ is holomorphic and 
\begin{align}
  \frac{d}{dz}\inner{\Phi}{\Psi(z)}  = \inner{\Phi}{\Xi_1(z)}, \qquad
  \frac{d^2}{dz^2}\inner{\Phi}{\Psi(z)}  = \inner{\Phi}{\Xi_2(z)}.
\end{align}
On the other hand, we have the bound
\begin{align}
  \norm{\Xi_2(z)} 
 & \leq \frac{1}{2} \norm{ \Delta^*(K''(z)) \Psi(z)} 
       + \frac{1}{4} \norm{ (\Delta^*(K'(z)))^2 \Psi(z)} \\
 & \leq \frac{1}{2} \norm{K''(z)}_2  \norm{ (\Nb+2) \Psi(z) } 
   + \frac{1}{4} \norm{K'(z)}_2^2\norm{ (\Nb+4)^2 \Psi(z)}.
\end{align}
This estimate, together with Theorem \ref{thm:GS-prop2}, implies that $\norm{\Xi_2(z)}$ is bounded on compact subsets of $U$. 
Therefore, by Lemma~\ref{lem:criterion-holomorphy}, $\Psi(z)$ is holomorphic.
\end{proof}

Next, for a given $K \in \cS_2(\sH)$ with $K=K^\mathrm{T}$ and $\norm{K}<1$, 
we show that there exists a family of creation and annihilation operators 
such that $\exp(-\tfrac{1}{2}\Delta^*(K))\Omega$ serves as the vacuum vector, 
that is, it is annihilated by the corresponding annihilation operators.
\begin{Proposition}\label{prop:C_CCR}
  Let $K\in \cB(\sH)$ satisfy $K=K^\mathrm{T}$ and $\norm{K}<1$, and set
  $\gamma \coloneqq  1-KK^*$. Then the pair
  $(\gamma^{-1/2}, K^*\gamma^{-1/2})$ belongs to $\mathfrak{Sp}$.
  In particular, if we define
  \begin{align}
    C(f) \coloneqq  \ovl{A(\gamma^{-1/2}f) + A^*(JK^*\gamma^{-1/2}f)},
    \qquad f\in \sH,   \label{def:C(f)}
  \end{align}
  then $\{C(f),C^*(f) \mid f\in \sH\}$ satisfies CCR on $\dom(\Nb)$.
\end{Proposition}
\begin{proof}
Set $X\coloneqq \gamma^{-1/2}$ and $Y\coloneqq  K^*\gamma^{-1/2}$. We show that these satisfy
the conditions in \eqref{symp_conditions}.  
The first identity follows directly from the definition:
\begin{align}
  X^*X-Y^*Y
    = \gamma^{-1/2}\gamma^{-1/2}
      -\gamma^{-1/2}KK^*\gamma^{-1/2}
    = \one.
\end{align}
Next, from $K=K^\mathrm{T}$ we obtain $JK^*=KJ$. Hence
\begin{align}
  X^*JYJ-Y^*JXJ
    = \gamma^{-1/2}JK^*\gamma^{-1/2}J
      -\gamma^{-1/2}KJ\gamma^{-1/2}J
    = 0,
\end{align}
which gives the second identity. Furthermore,
\begin{equation}
  KJ\gamma = KJ(1-KK^*) = K(1-K^*K)J = \gamma KJ,
\end{equation}
so $KJ$ commutes with $\gamma$, and therefore also with $\gamma^{-1}$.
Using this, we obtain the third identity:
\begin{align}
  XX^*-JYY^*J
    &= \gamma^{-1}-JK^*\gamma^{-1}KJ \\
    &= \gamma^{-1}-JK^*KJ\gamma^{-1} \\
    &= \one.
\end{align}
The fourth identity follows from the first three, which is a standard fact.
\end{proof}

\begin{Remark}
Given $(X,Y)\in \mathfrak{Sp}$, a representation of the CCR
$\{B(f),B^*(f)\mid f\in\sH\}$ is defined by \eqref{def:B(f)}.
On the other hand, starting from $(X,Y)\in \mathfrak{Sp}$ we define
$K_1\coloneqq  JYX^{-1}J$, and setting $K=K_1$, we construct a representation of the CCR
$\{C(f),C^*(f)\mid f\in\sH\}$ by the method described above.
We now clarify the relationship between these two representations.

In this case we have
$\gamma=1-K_1K_1^*=(XX^*)^{-1}$, and hence $\gamma^{-1/2}=|X^*|$.
Therefore,
\begin{align}
  C(f) = \ovl{ A(|X^*|f) + A^*(K_1J|X^*|f) }.
\end{align}
Furthermore, taking the polar decomposition $X=|X^*|V$, we see that $X$
is invertible and hence $V$ is unitary. Using this we obtain
\begin{align}
  C(f) &= \ovl{A(XV^*f)+A^*(K_1JXV^*f)}
       = \ovl{A(XV^*f)+A^*(JYV^*f)} \\
       &= B(V^*f).
\end{align}
Thus, $\{C(f),C^*(f) \mid f\in\sH\}$ is naturally unitarily equivalent to
$\{B(f),B^*(f) \mid f\in\sH\}$.
\end{Remark}

\medskip

The next theorem shows that the vector $e^{-\frac{1}{2}\Delta^*(K)}\Omega$
serves as an (unnormalized) vacuum vector
for the representation $\{C(f),C^*(f)\mid f\in\sH\}$.
\begin{Theorem}\label{thm:C_EXP}
Let $K\in\cS_2(\sH)$ satisfy $K=K^\mathrm{T}$ and $\norm{K}<1$.
Let $C(f)$ be the operator defined by \eqref{def:C(f)}. Then
\begin{align}
  C(f)e^{-\frac{1}{2}\Delta^*(K)}\Omega = 0 \qquad (f\in\sH)
\end{align}
holds.
In particular,  there exists a unique regular Bogoliubov transformation
$\cU$ on $\Fb(\sH)$ such that
\begin{align}
  \cU C(f)\cU^* = A(f) \qquad (f\in\sH)
\end{align}
and 
\begin{align}
 \cU e^{-\frac{1}{2}\Delta^*(K)}\Omega = c\Omega 
 \qquad (c=\norm{e^{-\frac{1}{2}\Delta^*(K)}\Omega} ).
\end{align}
\end{Theorem}
\begin{proof}
From the commutation relation
$[A(\gamma^{-1/2}f),\Delta^*(K)] = 2A^*(KJ\gamma^{-1/2}f)$,
we obtain
\begin{align}
  A(\gamma^{-1/2}f) e^{-\frac{1}{2}\Delta^*(K)}\Omega
  = - e^{-\frac{1}{2}\Delta^*(K)} A^*(KJ\gamma^{-1/2}f) \Omega.
\end{align}
Therefore, $C(f)e^{-\frac{1}{2}\Delta^*(K)}\Omega = 0$ holds, where we used the 
identity $KJ=JK^*$.

Since $Y=K^*\gamma^{-1/2}$ is a Hilbert--Schmidt operator,
the existence and uniqueness of the Bogoliubov transformation $\cU$
follows from the well-known general theory (see, for example, \cite{Rui78}).
\end{proof}

As an application of Theorem~\ref{thm:C_EXP}, we obtain the following result.
\begin{Theorem}\label{thm:analytic_vectors}
Let $K\in \cS_2(\sH)$ satisfy $K=K^\mathrm{T}$ and $\norm{K}<1$.
Then, for any $f,g\in\sH$ and any $n=0,1,2,\cdots$, we have
\begin{align}
 \norm{ (A(f)+A^*(g))^n e^{-\frac{1}{2}\Delta^*(K)}\Omega} 
 \leq 2^n \sqrt{n!}\,
 (\norm{\gamma^{-1/2}f}+\norm{\gamma^{-1/2}g})^n\,
 \norm{e^{-\frac{1}{2}\Delta^*(K)}\Omega}.
\end{align}
In particular, the series
\begin{align}
\exp(A(f)+A^*(g)) e^{-\frac{1}{2}\Delta^*(K)}\Omega 
 \coloneqq  \slim_{N\to \infty} \sum_{n=0}^N \frac{1}{n!}
 (A(f)+A^*(g))^n e^{-\frac{1}{2}\Delta^*(K)}\Omega
\end{align}
converges absolutely.
\end{Theorem}

\begin{proof}
By Theorem \ref{thm:C_EXP}, we have
$\cU e^{-\frac{1}{2}\Delta^*(K)}\Omega = c\Omega$.
Applying \eqref{eq:A(f)B(f)} with $(X,Y)=(\gamma^{-1/2},K^*\gamma^{-1/2})$
(see Proposition \ref{prop:C_CCR}) and \eqref{def:C(f)}, we have 
\begin{align}
A(f) = \overline{C(\gamma^{-1/2}f)-C^*(\gamma^{-1/2}KJf)}, \qquad 
A^*(g) = \overline{C^*(\gamma^{-1/2}g)-C(\gamma^{-1/2}KJg)}.
\end{align}
Hence,
\begin{align}
  A(f)+A^*(g) = C(F) + C^*(G) = \cU^*(A(F)+A^*(G))\cU
\end{align}
holds on $\Ffin(\sH)$, where 
\begin{align}
  F \coloneqq  \gamma^{-1/2}f - \gamma^{-1/2}KJg, \qquad 
  G \coloneqq  \gamma^{-1/2}g - \gamma^{-1/2}KJf.
\end{align}
Consequently, we have
\begin{align}
   \norm{ (A(f)+A^*(g))^n e^{-\frac12 \Delta^*(K)}\Omega}
 & = \norm{\cU^* (A(F)+A^*(G))^n \cU e^{-\frac12 \Delta^*(K)} \Omega} \\
 & = c\norm{(A(F)+A^*(G))^n \Omega}.
\end{align}

Using the standard estimates for creation and annihilation operators
and a simple induction argument, we obtain
\begin{align}
  \norm{ (A(F)+A^*(G))^n \Omega}
  \leq 2^n \max\{\norm{F},\norm{G}\}^n \sqrt{n!}.
\end{align}
Since $\norm{K}<1$ and $KJ$ commutes with $\gamma^{-1/2}$, we have
\begin{align}
  \norm{F},\, \norm{G}
  \leq \norm{\gamma^{-1/2}f} + \norm{\gamma^{-1/2}g}.
\end{align}
Therefore,
\begin{align}
  \norm{ (A(f)+A^*(g))^n e^{-\frac{1}{2}\Delta^*(K)}\Omega} 
  \leq  2^n c(\norm{\gamma^{-1/2}f}+\norm{\gamma^{-1/2}g})^n \sqrt{n!}
\end{align}
holds.
\end{proof}

\begin{Lemma}\label{lem:holo_AN}
Let $U\subset \CC$ be an open set, and let $f_j(z)\in \sH$ $(j=1,\cdots,n)$
be holomorphic functions on $U$.
Then, for any constant $c>0$, the operator-valued function
 \begin{align}
   A^{\sharp_1}(f_1(z^{\sharp_1}))\cdots  A^{\sharp_n}(f_n(z^{\sharp_n}))(\Nb+c)^{-n/2}
 \end{align}
is holomorphic on $U$ as a $\cB(\Fb(\sH))$-valued function.
Here, $A^\sharp$ denotes either $A$ or $A^*$.
Correspondingly, for $z\in\CC$, $z^\sharp$ denotes $\bar z$ or $z$, respectively.
\end{Lemma}
\begin{proof}
First we consider the case $n=1$. Using the bound
$\norm{A^*(f)(\Nb+1)^{-1/2}} \leq \norm{f}$ we obtain
\begin{align}
&\left\| \frac{A^*(f(z+h))(\Nb+1)^{-1/2} -  A^*(f(z))(\Nb+1)^{-1/2}}{h} 
       -  A^*(f'(z))(\Nb+1)^{-1/2}\right\|  \\
& = \left\| A^*\left(\frac{f(z+h)-f(z)}{h}-f'(z)\right) (\Nb+1)^{-1/2} \right\| \\
& \leq \left\| \frac{f(z+h)-f(z)}{h}-f'(z)\right\| 
       \to 0 \qquad (h\to 0).
\end{align}
Hence $A^*(f(z))(\Nb+1)^{-1/2}$ is holomorphic.
The same argument shows that $z\mapsto A(f(\bar{z})) (\Nb+1)^{-1/2}$ is holomorphic.

Next, let $n\ge2$. By the case $n=1$, each factor
$A^{\sharp_j}(f_j(z^{\sharp_j}))(\Nb+n)^{-1/2}$ defines a holomorphic
$\cB(\Fb(\sH))$-valued function on $U$, and hence so does their product
\begin{equation}
\prod_{j=1}^n A^{\sharp_j}(f_j(z^{\sharp_j}))(\Nb+n)^{-1/2}.
\end{equation}
Using the canonical commutation relations, we can move the factors
$(\Nb+n)^{-1/2}$ to the right and rewrite the above product in the form
\begin{equation}
A^{\sharp_1}(f_1(z^{\sharp_1}))\cdots A^{\sharp_n}(f_n(z^{\sharp_n}))
(\Nb+c_1)^{-1/2}\cdots(\Nb+c_n)^{-1/2},
\end{equation}
with suitable constants $c_j>0$.
Finally, for any $c>0$ the operator
\begin{equation}
(\Nb+c)^{n/2}(\Nb+c_1)^{-1/2}\cdots(\Nb+c_n)^{-1/2}
\end{equation}
is bounded, which yields the claim.
\end{proof}

\begin{Theorem}\label{thm:1_2_inner}
Suppose that $K\in \cS_2(\sH)$ and $\norm{K}<1$.
Then
\begin{align}
 \inner{e^{A^*(f)}\Omega}{e^{-\frac{1}{2}\Delta^*(K)}\Omega} 
 = e^{-\inner{f}{KJf}/2}, \qquad f\in\sH.
\end{align}
\end{Theorem}
\begin{proof}
First we assume that $K=K^\mathrm{T}$ and set $c\coloneqq \inner{f}{KJf}$.
Using the commutation relation 
$[A(f),e^{-\frac{1}{2}\Delta^*(K)}] = -A^*(KJf)e^{-\frac{1}{2}\Delta^*(K)}$, which holds on 
$\Ffin(\sH)$,
 we obtain
\begin{align}
   \inner{\Omega}{A(f)^{2n} e^{-\frac{1}{2}\Delta^*(K)}\Omega}  
& = -\inner{\Omega}{A(f)^{2n-1} A^*(KJf) e^{-\frac{1}{2}\Delta^*(K)}\Omega}  \\
& = -(2n-1)c\inner{\Omega}{A(f)^{2n-2}e^{-\frac{1}{2}\Delta^*(K)}\Omega}  \\
& = (2n-1)!! (-c)^{n},
\end{align}
where $(-1)!!=1$. Thus, 
\begin{align}
  \inner{e^{A^*(f)}\Omega}{e^{-\frac{1}{2}\Delta^*(K)}\Omega} 
  = \sum_{n=0}^\infty \frac{(2n-1)!!(-c)^{n}}{(2n)!} 
  = \sum_{n=0}^\infty \frac{(-c)^{n}}{(2n)!!} 
  = \sum_{n=0}^\infty \frac{(-c)^{n}}{2^n n!} = e^{-c/2}.
\end{align}
If $K\neq K^\mathrm{T}$, we note that $\Delta^*(K)=\Delta^*((K+K^\mathrm{T})/2)$.
Since 
\begin{align}
  \inner{f}{K^\mathrm{T}Jf}=\inner{f}{JK^*f}=\inner{K^*f}{Jf}=\inner{f}{KJf},  
\end{align}
the same computation yields the desired result.
\end{proof}

\section{Some integral inequalities related to trace class operators}
\label{AppenD}

Let $T$ be an injective non-negative self-adjoint operator acting on a Hilbert space $\sH$.
\begin{Theorem}\label{thm:trace_est1} 
Let $A \in \cS_1(\sH)$. Then, the Bochner integral
\begin{align}
  B \coloneqq  \int_0^\infty T^{1/2}(T^2+t^2)^{-1} A T^{1/2} (T^2+t^2)^{-1}  \, t^2dt   \label{0331_3}
\end{align}
converges absolutely in $\cS_1(\sH)$, and the inequality
\begin{align}
  \norm{B}_1 
  \leq \int_0^\infty \norm{ T^{1/2}(T^2+t^2)^{-1} A T^{1/2} (T^2+t^2)^{-1}}_1  \, t^2dt
  \leq \frac{\pi}{4} \norm{A}_1.   \label{0331_4}
\end{align}
holds.
\end{Theorem}

\begin{proof}
For $t>0$, since $T^{1/2}(T^2+t^2)^{-1}\in \cB(\sH)$, we have
 \begin{align}
   T^{1/2}(T^2+t^2)^{-1} A T^{1/2} (T^2+t^2)^{-1}  \in \cS_1(\sH).
 \end{align}
Since this is continuous in $t \in (0,\infty)$, to show that the integral converges, 
it suffices to prove that its trace norm is integrable with respect to $t^2dt$.
Since $A\in \cS_1(\sH)$, there exist positive numbers $\mu_n>0$ and orthonormal
systems $\{u_n\}_n$ and $\{v_n\}_n$ such that
\begin{align}
  A = \sum_{n=1}^N \mu_n \ketbra{u_n}{v_n}, \qquad \norm{A}_1 = \sum_{n=1}^N \mu_n <\infty,  \label{0331_5}
\end{align}
where $N\leq \infty$.
By applying the triangle inequality and the Cauchy--Schwarz inequality, we obtain:
\begin{align}
& \int_0^\infty \norm{ T^{1/2}(T^2+t^2)^{-1} A T^{1/2} (T^2+t^2)^{-1} }_1 \, t^2dt \\
& \leq  \int_0^\infty \sum_{n=1}^N \mu_n\norm{ T^{1/2}(T^2+t^2)^{-1} \ketbra{u_n}{v_n} T^{1/2} (T^2+t^2)^{-1} }_1 \, t^2dt \\
& =  \int_0^\infty \sum_{n=1}^N \mu_n\norm{ T^{1/2}(T^2+t^2)^{-1} u_n} \norm{T^{1/2} (T^2+t^2)^{-1} v_n} \, t^2dt \\
& \leq \sum_{n=1}^N \mu_n \left( \int_0^\infty \norm{ T^{1/2}(T^2+t^2)^{-1} u_n}^2 \, t^2 dt \right)^{1/2} 
 \left( \int_0^\infty \norm{ T^{1/2}(T^2+t^2)^{-1} v_n}^2 \, t^2dt  \right)^{1/2}.  \label{0331_2}
\end{align}
Noting that $\int_0^\infty T(T^2+t^2)^{-2}\, t^2dt = \pi/4$, we have
\begin{align}
  \eqref{0331_2} = \frac{\pi}{4}\sum_{n=1}^N \mu_n \norm{u_n} \norm{v_n} 
                 = \frac{\pi}{4}\norm{A}_1.
\end{align}
This implies that \eqref{0331_3} converges in $\cS_1(\sH)$ and that \eqref{0331_4} holds.
\end{proof}

\begin{Theorem}\label{thm:trace_est2}
 Let $A\in \cS_1(\sH)$. Then it holds that
 \begin{align}
   \int_0^\infty \norm{T^{1/2}(T^2+t^2)^{-1} A}_1^2 \, t^2 dt 
 \leq \frac{\pi}{4} \norm{A}_1^2.
 \end{align}
\end{Theorem}

\begin{proof}
From \eqref{0331_5} and the fact that $\norm{v_n}=1$, it follows that
  \begin{align}
  \norm{T^{1/2}(T^2+t^2)^{-1} A}_1
 \leq \sum_{n=1}^N \mu_n \norm{T^{1/2}(T^2+t^2)^{-1} u_n}.
  \end{align}
Thus, following a similar argument as in the proof of Theorem \ref{thm:trace_est1}, one obtains 
\begin{align}
&   \int_0^\infty \norm{T^{1/2}(T^2+t^2)^{-1} A}_1^2 \, t^2 dt  \\
& \leq \sum_{n,m=1}^N \mu_n \mu_m\int_0^\infty \norm{T^{1/2}(T^2+t^2)^{-1} u_n}
       \norm{T^{1/2}(T^2+t^2)^{-1} u_m} \, t^2dt \\
& \leq  \frac{\pi}{4}\sum_{n,m=1}^N\mu_n\mu_m = \frac{\pi}{4}\norm{A}_1^2.
\end{align}
\end{proof}

\begin{Theorem}\label{thm:trace_est3}
 Let $A_1,A_2\in \cS_1(\sH)$. 
 Assume that $Q_t \in \cB(\sH)$ is weakly measurable in $t\in (0,\infty)$,
i.e., $t\mapsto\inner{u}{Q_tv}$ is Borel measurable for all $u,v\in\sH$, 
and that $\sup_{t>0}\norm{Q_t}<\infty$.
Then, the Bochner integral
 \begin{align}
 B \coloneqq  \int_0^\infty T^{1/2} (T^2+t^2)^{-1} A_1 Q_t A_2 T^{1/2} (T^2+t^2)^{-1} \, t^2dt
 \end{align}
converges absolutely in $\cS_1(\sH)$, and the inequality 
\begin{align}
  \norm{B}_1
  & \leq \int_0^\infty \norm{ T^{1/2} (T^2+t^2)^{-1} A_1 Q_t A_2 T^{1/2} (T^2+t^2)^{-1}}_1  \, t^2dt \\
  & \leq \frac{\pi}{4} \left(\sup_{t>0}\norm{Q_t}\right) \norm{A_1}_1 \norm{A_2}_1
\end{align}
holds.
\end{Theorem}

\begin{proof}
From the assumptions, $T^{1/2} (T^2+t^2)^{-1} A_1 Q_t A_2 T^{1/2} (T^2+t^2)^{-1}$
is measurable in the trace norm.
By sub-multiplicativity, we have
  \begin{align}
&  \norm{T^{1/2} (T^2+t^2)^{-1} A_1 Q_t A_2 T^{1/2} (T^2+t^2)^{-1} }_1 \\
& \leq \norm{T^{1/2} (T^2+t^2)^{-1} A_1}_1 \norm{Q_t} \norm{A_2 T^{1/2} (T^2+t^2)^{-1} }.
  \end{align}
Integrating both sides and applying the Cauchy--Schwarz inequality, we obtain
\begin{align}
&  \int_0^\infty \norm{T^{1/2} (T^2+t^2)^{-1} A_1 Q_t A_2 T^{1/2} (T^2+t^2)^{-1} }_1 \, t^2 dt \\
& \leq \left(\sup_{t>0}\norm{Q_t}\right)
       \left( \int_0^\infty \norm{T^{1/2} (T^2+t^2)^{-1} A_1}_1^2 \, t^2 dt\right)^{1/2} \\
& \quad \times   \left( \int_0^\infty \norm{A_2 T^{1/2} (T^2+t^2)^{-1}}^2 \, t^2 dt\right)^{1/2}.
\end{align}
By using Theorem \ref{thm:trace_est2} to estimate the above expression, the conclusion follows.
\end{proof}

\section{Complex differentiation under the integral sign}\label{appenE}
For the reader's convenience, we recall a result on complex differentiation
 under the integral sign due to L.~Mattner~\cite{MR1823156}. 
This result is elementary but may not be widely known. 
By this theorem, one does not need to find an integrable dominating function for the derivative; 
it requires only the local uniform boundedness of the integrals with respect to the parameter.
\begin{Theorem}[{\cite{MR1823156}}] \label{thm:comp-diff}
Let $(X, \cB, \mu)$ be a $\sigma$-finite measure space, and let $U$ be a non-empty open subset of $\CC$. 
Suppose that a function $f : U \times X \to \CC$ satisfies the following conditions:
\begin{enumerate}
\item[\rm (i)] For each $z \in U$, the function $X \ni x \mapsto f(z,x) \in \CC$ is $\cB$-measurable.
\item[\rm (ii)] For each $x \in X$, the function $U \ni z \mapsto f(z,x) \in \CC$ is holomorphic.
\item[\rm (iii)] For any compact subset $K \subset U$, 
  \begin{equation}
      \sup_{z\in K} \int_X|f(z,x)| \, d\mu(x)  < \infty.
  \end{equation}
\end{enumerate}
Then, the function $F$ defined by
\begin{equation}
  F(z) \coloneqq  \int_X f(z,x) \, d\mu(x), \qquad z \in U
\end{equation}
is holomorphic on $U$. Furthermore, the following statements hold:
\begin{enumerate}
\item[\rm (1)] For each $n \in \NN$ and $z \in U$, the function $x \mapsto (\del^n f / \del z^n)(z,x)$ 
  is $\cB$-measurable and $\mu$-integrable.
\item[\rm (2)] For each $n \in \NN$ and $z \in U$, 
  \begin{equation}
      \frac{d^n F}{dz^n}(z) = \int_X \frac{\del^n f}{\del z^n} (z,x) \, d\mu(x).
  \end{equation}
\end{enumerate}
\end{Theorem}
\vspace*{10pt}
\noindent\textbf{Acknowledgments:}
This work was supported by JSPS KAKENHI (Grant Numbers JP16K17612, JP20K03628, JP23K25783, and JP24K06755).
This work was also supported by the Research Institute for Mathematical Sciences,
an International Joint Usage/Research Center located in Kyoto University.


\bibliographystyle{plain}
\bibliography{References}

\bigskip
\vfill 

\noindent
\begin{minipage}[t]{0.48\textwidth}
\raggedright
{\small
\textbf{Kota Imura}\\
Nagano Prefecture Fujimi High School,\\
Fujimi 3330, Nagano 399-0211, Japan\\
E-mail: \texttt{t30005206ik@g.nagano-c.ed.jp}

\medskip

\textbf{Shinnosuke Izumi}\\
Department of Mathematics,\\
Shinshu University,\\
Matsumoto, Nagano 390-8621, Japan\\
E-mail: \texttt{izumi@math.shinshu-u.ac.jp}
}
\end{minipage}%
\hfill
\begin{minipage}[t]{0.48\textwidth}
\raggedright
{\small
\textbf{Yasumichi Matsuzawa}\\
Department of Mathematics,\\
Faculty of Education, Shinshu University,\\
6-Ro, Nishi-nagano, Nagano 380-8544, Japan\\
E-mail: \texttt{myasu@shinshu-u.ac.jp}

\medskip

\textbf{Itaru Sasaki}\\
Department of Mathematics,\\
Shinshu University,\\
Matsumoto, Nagano 390-8621, Japan\\
E-mail: \texttt{isasaki@shinshu-u.ac.jp}
}
\end{minipage}

\end{document}